\documentclass[12pt, draftclsnofoot, onecolumn]{IEEEtran}
 \makeatletter

\newcommand{\Rmnum}[1]{\expandafter\@slowromancap\romannumeral #1@}
\makeatother
\usepackage{graphicx,cite,epsfig,amssymb,amsmath,multirow}
\usepackage{graphicx,amsmath,amssymb,amsfonts,cite}
\usepackage{mathrsfs}
\usepackage{multirow}
\usepackage{algorithm}
\usepackage{graphicx}
\usepackage{subfigure}
\usepackage{amsmath}
\usepackage{algorithmic}
\usepackage{float}
\usepackage{color}
\usepackage{bm}
\usepackage{stfloats}

\usepackage{amsmath,amsthm,amssymb,amsfonts}
\newtheorem{theorem}{Theorem}
\newtheorem{corollary}{Corollary}
\newtheorem{lemma}{Lemma}
\newtheorem{remark}{Remark}
\begin{document}
\title{Uplink Spectral Efficiency Analysis of Decoupled Access in Multiuser MIMO Communications}
\author{Ran Li, Kai Luo,~\IEEEmembership{~Member,~IEEE,} Tao Jiang, ~\IEEEmembership{Senior~Member,~IEEE}, and Shi Jin,~\IEEEmembership{~Member,~IEEE}\\
\thanks{This research was supported by the National Natural Science Foundation of China (NSFC) under Grant 61531011.

R. Li, K. Luo and T. Jiang are with Wuhan National Laboratory For Optoelectronics, School of Electronic Information and Communications, Huazhong University of Science and Technology, Wuhan 430074, P. R. China (e-mails: \{ranlee, kluo, taojiang\}@hust.edu.cn).

S. Jin is with the National Mobile Communications Research Laboratory, Southeast University, Nanjing 210096, P. R. China (e-mail: jinshi@seu.edu.cn).
}}

\maketitle

\begin{abstract}
  In a heterogeneous network consisting of macro base stations (MBSs) and small base stations (SBSs), the traditional cell association policy, i.e., coupled access (CA), is far from optimal, due to the significant difference between the coverage and transmit powers of MBSs and SBSs. Hence, users may choose to associate with different types of BSs in downlink (DL) and uplink (UL), i.e., decoupled access (DA), to enhance spectral efficiency. In this paper, DA in multiuser MIMO communications is investigated in terms of UL spectral efficiency. Firstly, we obtain the UL association probabilities. In contrast to the CA scenario, association probabilities for DA scenario only depend on the densities of BSs. Hence, DA allows UL and DL to be totally independent. Secondly, we derive lower bounds on the spectral efficiency. The lower bounds show that, different from CA, the UL spectral efficiency for DA scenario is irrelative with the transmit powers of BSs, which implies DA allows users to associate with any BSs that can achieve the highest UL spectral efficiency. Finally, the spectral efficiencies for DA and CA scenarios are compared via simulation results, where it can be concluded that
the spectral efficiency in multiuser MIMO systems is improved by DA.
\end{abstract}
\begin{keywords}
Heterogeneous networks, decoupled access, spectral efficiency, multiuser MIMO.
\end{keywords}
\section{Introduction}\label{section1}
In the quest for the ever increasing traffic demands, a growing number of base stations~(BSs), especially low-power small BSs (SBSs), are added to the conventional single-tier wireless cellular networks, leading to the evolution of current networks towards a multi-tier heterogeneous infrastructure \cite{5876496,6171992,6472194}.
In a heterogeneous deployment, while the existing macro BSs~(MBSs) provide full area coverage, various complementary SBSs, e.g., pico BSs and femto BSs, help offload MBSs and provide high traffic capacity as well as enhanced service experience \cite{astely2013lte,4623708,5483516}.
However, different types of BSs result in massive differences in transmit powers and hence coverage areas, introducing a major asymmetry between uplink (UL) and downlink (DL)~\cite{6476878}, i.e., the optimal BSs for a user in DL and UL may be different.
Meanwhile, symmetric traffic applications, e.g., video calls and social networking, call for high traffic demands in UL, leading to the increasing importance of improving UL performance \cite{6516885}.
Hence, current cell association in traditional cellular networks, which is, one user selects the same BS in UL as that in DL according to the maximum DL received power, i.e., coupled access (CA), is far from optimal in heterogeneous networks. Therefore, decoupled access (DA) for UL and DL is highly demanded for investigation \cite{6736746}.

DA allows access points in UL and DL to be different, contributing to better resource allocation between cells, based on channel conditions, service types and BS traffic loads~\cite{6736745}, thus resulting in enhanced UL performance.~Some works have been done on the analysis of DA~\cite{7037069,7003998,DBLP:journals/corr/SmiljkovikjEPBDGI14,7247166,7112544,7249179}, where~the locations of BSs and users were modeled as homogeneous Poisson point processes~(PPPs), which was described in~\cite{6171996,6524460,Andrews2011}.
The UL throughput of a DA system was studied in~\cite{7037069,7003998}.
The authors in \cite{7037069} investigated the performance of DA in a network consisting of a MBS and a SBS using a system level simulation tool Atoll, proving that gains can be achieved by DA in terms of UL throughput.
Then, the association probabilities and average throughput for users with DA and CA were analyzed in a heterogeneous network composed of a MBS tier and a SBS tier in \cite{7003998}.
Based on the association probabilities, DA was studied from the UL spectral efficiency perspective in \cite{DBLP:journals/corr/SmiljkovikjEPBDGI14,7247166}, where the superiority of DA over CA was demonstrated.
In~\cite{DBLP:journals/corr/SmiljkovikjEPBDGI14}, the spectral efficiency of a decoupled system was studied analytically for a homogeneous user domain and validated by a real-world simulation, while the spectral efficiency and energy efficiency were calculated in \cite{7247166} for a heterogeneous user domain, i.e., the transmit power of users associated to MBSs is higher than that of users associated to SBSs.
Moreover, besides the analysis of UL performance, the authors in \cite{7112544} also took DL performance into consideration and studied the joint UL and DL rate coverage, proving that DA leads to significant improvement in the joint rate coverage over the traditional CA in a $K$-tier heterogeneous network.
Furthermore, instead of taking the UL association decisions based only on the UL received power for DA scenario, the authors in \cite{7249179} proposed a cell association algorithm which extends the association criterion to include the cell load and backhaul capacity.

However, the above works consider a single user scenario where each user is served by a BS while each BS serves only a single user in one resource block.
In fact, a BS, especially a MBS with large antenna arrays, might serve multiple users simultaneously in a given resource block.
In this paper, the spectral efficiency analysis is considered in this multiuser scenario where each user is served by a BS while each BS serves multiple users in one resource block.
Different from the previous works, the key features of the work in this paper are:
\begin{itemize}
  \item In a MIMO system, precoders/detectors are required. The zero-forcing (ZF) precoder and detector are used according to \cite{5956512} and \cite{6082486}, since the performance of the ZF detector is better than that of the maximum-ratio combining detector and the complexity of the ZF detector is lower than that of the minimum mean squared error detector. Besides, we normalize each column of the precoder rather than normalize the whole matrix, since the former is shown to deliver a higher sum-power and sum-rate than the latter according to~\cite{6560005}.
  \item The transmit powers of BSs are assumed to be equally allocated among the users associated to them. By using the equal power allocation, cell associations with CA are analyzed for a comparison with the DA scenario.
\end{itemize}
Based on the two key features, we model the locations of BSs and users as independent homogeneous PPPs.
Then, the received signals in DL and UL are derived.
The contributions of the paper are as follows.

First, we obtain the association probabilities of users with DA and CA in the multiuser scenario.
According to the derived average received signal power and the association criterion, theoretical analysis of cell association is presented.
By applying the cumulative distribution function (CDF) and the corresponding probability density function (PDF) of the distance between one user and its nearest BS, the association probabilities for both DA and CA scenarios are derived.
Then, the association probabilities are simplified for a number of special cases, namely combinations of~(i) each BS serves at most one user, (ii) numbers of antennas equipped in any BSs are identical, and (iii) the transmit power of MBSs is the same as that of SBSs.
In the case that (i) and (ii) are taken, we provide the same results as those for single user scenario in the previous works.
Furthermore, according to the derived association probabilities, we present the following insights.
First, the raising transmit power of MBSs results in higher probability that users with CA will associate with MBSs, while the association probabilities of users with DA remain unchanged.
Second, with the increase of density of SBSs, users with both CA and DA are less likely to associate with MBSs in UL.
The insights are validated later by simulations.

Next, we develop the lower bounds on the spectral efficiency with both DA and CA.
Since the spectral efficiency is defined as the expectation of channel capacity normalized by system bandwidth where the signal-to-interference-plus-noise ratio~(SINR) is involved, based on the received signal in UL, we analyze the spectral efficiency with the derived UL SINR.
Moreover, the lower bounds are also derived for a special case when each BS serves at most one user.
In that particular case, it is further shown that the spectral efficiency of users with DA and CA are identical when BSs have equal numbers of antennas and equal transmit powers.
The analysis of the spectral efficiency are validated by simulations, where we observe that the lower bounds can describe the trends of the accurate spectral efficiency and provide tractable predictions of the ratio between the spectral efficiencies with DA and CA.
Then, a comparison between the spectral efficiencies for DA and CA scenarios is made by simulations, leading to the insight that the spectral efficiency of users with DA is much higher than that of users with CA over a broad range of parameter~values.

The rest of the paper is organized as follows. In Section \ref{section2}, the system model is described and the received signals in DL and UL are derived. In Section \ref{section3}, the association probabilities are calculated. Then the spectral efficiencies with DA and CA are analyzed in Section \ref{section4}. In Section~\ref{section5}, numerical results are conducted and a comparison between the spectral efficiencies with DA and CA is made. Finally, we conclude the paper in Section \ref{section6}.

 \emph{Notation:} Upper and lower case boldface letters are used to denote matrices and vectors, respectively. The conjugate transpose is represented by $(\cdot)^H$. The trace of a matrix is denoted as~${\rm tr}(\cdot)$. The expectation operator with respect to~$x$ is represented as $\mathbb{E}_x\left[\cdot\right]$. The sets of complex- and real-valued $N\times K$ matrices are denoted as $\mathbb{C}^{N\times K}$ and $\mathbb{R}^{N\times K}$, respectively. The real part of a complex number is represented by $\mbox{Re}\{\cdot\}$. The diagonal matrix is expressed by ${\rm diag}(\cdot)$ and the Euclidean norm is $\left\lVert\cdot\right\rVert$.
\section{System Model}\label{section2}
We consider a two-tier heterogeneous cellular network which consists of a macro cell tier and a small cell tier. The locations of MBSs, SBSs and users are modeled as homogeneous PPPs~$\Phi_M$, $\Phi_S$ and $\Phi_u$ with intensity $\lambda_M$,~$\lambda_S$ and~$\lambda_u$, respectively.
Each $v$BS, where $v\in\{M,S\}$, deploys~$L_v$ antennas with the total transmit power being $P_v$, while each user is equipped with one single antenna with the transmit power being $Q$.
The multiuser single connection scenario with DA is considered (see Fig. \ref{multi-user}). Note that there are $M$ MBSs and~$N$ SBSs in a certain area.
In this scenario, each BS could serve multiple users simultaneously, while each user is assumed to connect to its nearest MBS or SBS only.
Then, let $K_{m,n}$ represent the number of users whose nearest MBS is the~$m$th MBS ($\mbox{MBS}_m$) and nearest SBS is the $n$th SBS~($\mbox{SBS}_n$).
Furthermore, let $K_M$ denote the number of users whose nearest MBS is $\mbox{MBS}_m$ and~$K_S$ denote the number of users whose nearest SBS is $\mbox{SBS}_n$, we have
 \begin{eqnarray}
   K_M&=&\sum_{n=1}^NK_{m,n},\\
   K_S&=&\sum_{m=1}^MK_{m,n}.
 \end{eqnarray}

In UL or DL, there are $K_{M,m,n}^{(\cdot)}$ and $K_{S,m,n}^{(\cdot)}$ users associated to $\mbox{MBS}_m$ and $\mbox{SBS}_n$ among the~$K_{m,n}$ users, respectively, where $(\cdot)\!\!=\!D$ for DL and $(\cdot)\!\!=\!U$ for UL. Then, we have
\begin{equation}
  K_{M,m,n}^{(\cdot)}+K_{S,m,n}^{(\cdot)}= K_{m,n}.
\end{equation}
When the $k$th user in the set of $K_{m,n}$ users is associated to $\mbox{MBS}_m$ in UL or DL,
there are $K_M^{(\cdot)}$ users in total associated to $\mbox{MBS}_m$, given by
 \begin{equation}
   K_M^{(\cdot)}\!\!\!=\!K_{M,m,n}^{(\cdot)}\!\!+\!\!\displaystyle{\sum_{i=1\atop{i\neq n}}^{N}}K_{M,m,i}^{(\cdot)},\; 1\!\leq\!\! K_{M}^{(\cdot)}\!\!\leq\! K_M,K_{M}^{(\cdot)}\!\!\le\! L_M,
\end{equation}
where $K_{M,m,i}^{(\cdot)}$ denotes the number of users associated to $\mbox{MBS}_m$ while their nearest SBS is $\mbox{SBS}_i$~($i\!\!\neq \!\!n$).
Similarly, when the $k$th user is associated to $\mbox{SBS}_n$ in UL or DL, there are $K_S^{(\cdot)}$ users in total associated to $\mbox{SBS}_n$, where
\begin{equation}
  K_S^{(\cdot)}\!\!\!=\!K_{S,m,n}^{(\cdot)}\!+\!\!\sum_{j=1\atop{j\neq m}}^MK_{S,j,n}^{(\cdot)},\;1\!\le\! K_S^{(\cdot)}\!\le\! K_S,\,K_S^{(\cdot)}\!\le\! L_S,
\end{equation}
in which $K_{S,j,n}^{(\cdot)}$ is the number of users associated to $\mbox{SBS}_n$ while their nearest MBS is $\mbox{MBS}_j$~($j\!\!\neq~\!\!\!m$).

 Based on such a model, the association cases and probabilities of the $k$th user in the set of~$K_{m,n}$ users are analyzed after the derivation of the received signals in DL and UL.
\begin{figure}
\centering
\includegraphics[width=3.5in]{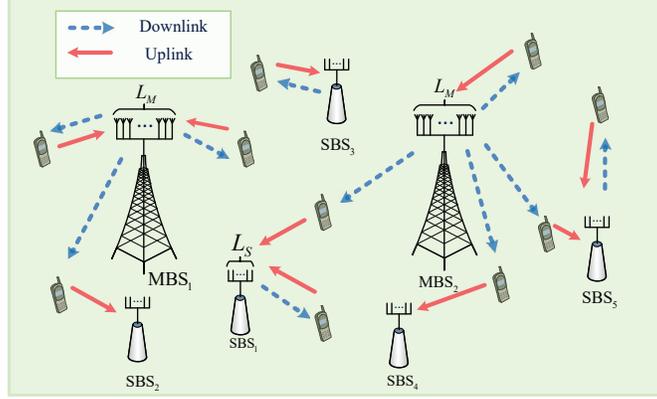}
\caption{The illustration of DA in the multiuser single connection scenario}
\label{multi-user}
\end{figure}
\subsection{Received signal in DL}
Let $\mathbf{y}_{v}^{(D)}$ denote the effective received signal from $v$BS where $v\in\{M,S\}$, i.e.,
\begin{equation}\label{1}
  \mathbf{y}_{v}^{(D)}=\mathbf{P}_{v}^{1/2}\,\mathbf{R}_{v}^{1/2}\,\mathbf{H}_{v}\mathbf{U}_{v}\mathbf{z}_{v}\in \mathbb{C}^{K_v^{(D)}\times 1}.
\end{equation}
Note that the channel matrix $\mathbf{H}_{v}\!\in\! \mathbb{C}^{K_v^{(D)}\times L_v}$ represents the channels between $v$BS and the users associated to it.
For simplicity and without loss of generality, each component of~$\mathbf{H}_{v}$ is assumed to be an independent and identical distributed (i.i.d.) complex Gaussian random variable with zero mean and unit variance.
Besides, the diagonal matrix $\mathbf{R}_{v}\in\mathbb{C}^{K_v^{(D)}\times K_v^{(D)}}$ represents the path loss. The $k$th entry in $\mathbf{R}_{v}$ is denoted as~$\left[\mathbf{R}_{v}\right]_{kk}\!\!=\!\!r_{v,k}^{-\alpha}$,
where $r_{v,k}$ is the distance from the $k$th user to its tagged $v$BS and $\alpha$ is the path loss exponent.
Assume that the total transmit power of $v$BS is equally allocated among the~$K_v^{(D)}$ users associated to it. Then, we have
\begin{equation}
 \mathbf{P}_{v}\!=\!{\rm diag}\!\left(\!\frac{P_{v}}{K_v^{(D)}},...,\frac{P_{v}}{K_v^{(D)}}\!\right)\in \mathbb{R}^{K_v^{(D)}\times K_v^{(D)}}.
\end{equation}
Furthermore, the precoded transmitted signal is denoted as $\mathbf{U}_{v}\mathbf{z}_{v}\in \mathbb{C}^{L_v\times 1}$ where~$\mathbf{U}_{v}\!\in\! \mathbb{C}^{L_{v}\!\times\! K_v^{(D)}}$ is a precoder and $\mathbf{z}_{v}$ is the $K_v^{(D)}\!\!\!\times\! 1$ data symbol vector.
Note that the~$k$th entry of $\mathbf{z}_{v}$ satisfies~$\mathbb{E}\left[\left\lvert{z}_{v,k}\right\rvert^2\right]=1$ and the term $\mathbf{U}_{v}\mathbf{z}_{v}$ is normalized for the $k$th user.
Thus, the energy constraint on the~$k$th column of $\mathbf{U}_{v}$ is obtained, i.e.,~$\mathbb{E}\left[\left\lVert\mathbf{u}_{v,k}\right\rVert^2\right]=1$, which will be used later to normalize the transmit power.

\subsection{Received signal in UL}
Let $\mathbf{x}_{v}\!\!\in\! \mathbb{C}^{K_v^{(U)}\!\times\! 1}$ represent the normalized transmitted signal across the $K_v^{(U)}$ users associated to $v$BS with $\mathbb{E}\!\left[\left\lvert{x}_{v,k}\right\rvert^2\right]\!=\!1$. Then, the received signal at $v$BS, i.e., $\mathbf{y}_{v}^{(U)}\!\!\in\! \mathbb{C}^{L_v\!\times\! 1}$, is given as
\begin{equation}\label{UL}
  \mathbf{y}_{v}^{(U)}=\mathbf{G}_{v}\mathbf{Q}_{v}^{1/2}\,\mathbf{x}_{v}+\hspace{-0.3cm}\sum_{{j\in\Phi_u\backslash \left\{\!K_v^{(U)}\!\right\}}}\hspace{-0.3cm}\sqrt{Qr_{v,j}^{-\alpha}}\mathbf{h}_{v,j}x_{j}+\mathbf{n}_{v}^{(U)}.
\end{equation}
The first term $\mathbf{G}_{v}\mathbf{Q}_{v}^{1/2}\,\mathbf{x}_{v}$ represents the received signal of~$v$BS from users associated to it where $\mathbf{G}_{v}\in~\mathbb{C}^{L_v\times K_v^{(U)}}$ is the channel matrix with each component being an i.i.d. complex Gaussian random variable with zero mean and unit variance.
Moreover, the diagonal matrix $\mathbf{Q}_{v}$ with~$\left[\mathbf{Q}_{v}\right]_{kk}=~Qr_{v,k}^{-\alpha}$ is the received signal power of~$v$BS from users associated to it.
The second term in (\ref{UL}) is the interference from all the users associated to other BSs, where~$r_{v,j}$ is the distance from~$v$BS to the $j$th user which is not associated to it.
Besides, the symbol $\mathbf{h}_{v,j}$ is an i.i.d. complex Gaussian distributed vector with zero mean and unit variance, representing the channel between $v$BS and the $j$th user.
Furthermore, the transmit signal of the $j$th user is denoted as $x_j$. Finally, the third term $\mathbf{n}_{v}^{(U)}$ is a vector of additive white Gaussian noise at $v$BS with $\mathbf{n}_{v}^{(U)}\sim \mathcal{CN}(0,\sigma^2\mathbf{I}_{L_v})$.

Next, the average received signal powers will be obtained by using the derived received signals in DL (Section \ref{section2}-A) and UL (Section \ref{section2}-B). Further, expressions for the association probabilities with CA and DA will be developed. Meanwhile, with the received signal in UL, the UL SINR for both DA and CA scenarios can be obtained. Then, based on the SINR and the association probabilities, the average spectral efficiency will be analyzed.
\section{Association Probability Analysis}\label{section3}
In this section, the UL association probabilities for both DA and CA scenarios are analyzed. It is assumed that the BSs and users have perfect channel state information. We start with calculating the average received power at the $k$th user in DL and that at its tagged BS in UL. Based on these results, the association probabilities are derived.
\subsection{Average received signal power}
In this subsection, the average received signal powers in both DL and UL are derived, which will be utilized to obtain association probabilities with CA and DA.
\subsubsection{Average received signal power in DL}
\

When the $k$th user is associated to $v$BS, there are $K_v^{(D)}$ users in total associated to $v$BS. Then, ZF precoder, given as follows, is used to eliminate the interference caused by other users associated to $v$BS,
\begin{equation}\label{ZF}
  \mathbf{W}_{v}=\mathbf{H}_{v}^H(\mathbf{H}_{v}\mathbf{H}_{v}^H)^{-1}\in\mathbb{C}^{K_v^{(D)}\times L_v}.
\end{equation}
Let $\mathbf{W}_{v}$ be written as
\begin{equation}\label{columns}
  \mathbf{W}_{v}=\left[\mathbf{w}_{v,1},...,\mathbf{w}_{v,k},...,\mathbf{w}_{v,K_v^{(D)}}\right].
\end{equation}
Then, by using the equal transmit power normalization \cite{6560005}, the normalized precoder $\mathbf{U}_{v}$ is given as
\begin{equation}\label{15}
  \mathbf{U}_{v}=\left[\frac{\mathbf{w}_{v,1}}{\left\lVert\mathbf{w}_{v,1}\right\rVert},...,\frac{\mathbf{w}_{v,k}}{\left\lVert\mathbf{w}_{v,k}\right\rVert},...,\frac{\mathbf{w}_{v,K_v^{(D)}}}{\left\lVert\mathbf{w}_{M,K_v^{(D)}}\right\rVert}\right].
\end{equation}
By substituting (\ref{15}) into (\ref{1}), the $k$th element of $\mathbf{y}_{v}^{(D)}$ is obtained, i.e.,
\begin{equation}
  {y}_{v,k}^{(D)}=\frac{P_{v}}{K_v^{(D)}}r_{v,k}^{-\alpha/2}\frac{1}{\left\lVert\mathbf{w}_{v,k}\right\rVert}x_{v,k}.
\end{equation}
Hence, the received signal power of the $k$th user associated to $v$BS, i.e., ${\rm S}_{v,k}^{(D)}$, is derived as
\begin{eqnarray}\label{17}
 {\rm S}_{v,k}^{(D)}=\frac{ P_{v}/K_v^{(D)}}{\left\lVert\mathbf{w}_{v,k}\right\rVert^2}r_{v,k}^{-\alpha}.
\end{eqnarray}
Thus, the average received signal power is given in the following lemma.
\begin{lemma}\label{lemma01}
The average received signal power is calculated~as
\begin{eqnarray}\label{SM}
 \mathbb{E}\left[{\rm S}_{v,k}^{(D)}\right]=\frac{P_{v}\left(L_v-K_v^{(D)}+1\right)}{K_v^{(D)}}r_{v,k}^{-\alpha}.
\end{eqnarray}
\end{lemma}
\begin{proof}
According to (\ref{17}), we have
\begin{eqnarray}\label{erlang}
 \mathbb{E}\left[{\rm S}_{v,k}^{(D)}\right]\hspace{-0.2cm}&=&\hspace{-0.2cm}\frac{P_{v}}{K_v^{(D)}}\mathbb{E}\left[\frac{1}{(\mathbf{H}_{v}\mathbf{H}_{v}^H)^{-1}_{kk}}\right]r_{v,k}^{-\alpha},
\end{eqnarray}
where the term $\frac{1}{(\mathbf{H}_{v}\mathbf{H}_{v}^H)^{-1}_{kk}}$ has an Erlang distribution with~$\frac{1}{(\mathbf{H}_{v}\mathbf{H}_{v}^H)^{-1}_{kk}}\!\!\sim\!\!\emph{Erlang}\!\left(L_v\!-\!K_v^{(D)}\!+\!1,\!1\right)$~\cite{6082486}, which means
\begin{equation}\label{20}
  \mathbb{E}\left[\frac{1}{(\mathbf{H}_{v}\mathbf{H}_{v}^H)^{-1}_{kk}}\right]=L_v-K_v^{(D)}+1.
\end{equation}
By substituting (\ref{20}) into (\ref{erlang}), the desired result is derived.
\end{proof}
Based on the average received signal power in (\ref{SM}), the UL association probabilities with CA will be derived.
\subsubsection{Average received signal power in UL}
\

Here, the received signal in UL after using linear detector is obtained. Based on the received signal, the average received signal power and SINR are derived. Then, these results will be utilized to explore the UL association probabilities for DA and analyze the spectral efficiencies with both DA and CA.

By using the ZF linear detector, the received signal of $v$BS is separated into streams, given~by
\begin{equation}\label{26}
  \mathbf{s}_{v}=\mathbf{A}_{v}^H\mathbf{y}_{v}^{(U)}\in\mathbb{C}^{K_v^{(U)}\times 1},
\end{equation}
where $\mathbf{A}_{v}\in\mathbb{C}^{L_v\times K_v^{(U)}}$ depends on the channel $\mathbf{G}_{v}$, i.e.,
\begin{equation}\label{21}
  \mathbf{A}_{v}=\mathbf{G}_{v}(\mathbf{G}_{v}^H\mathbf{G}_{v})^{-1}.
\end{equation}
By substituting (\ref{UL}) and (\ref{21}) into (\ref{26}), we derive $s_{{v},k}$ (the $k$th element of $\mathbf{s}_{v}$), which is the received signal of $v$BS from the $k$th user, as follows, where $x_{{v},k}$ is the $k$th row of $\mathbf{x}_{{v}}$,
\begin{equation}\label{UL_MBS}
  s_{{v}\!,k}\!\!=\!\!\sqrt{\!Qr_{v\!,k}^{-\alpha}}x_{{v}\!,k}\!+\hspace{-0.4cm}\!\sum_{{j\in\Phi_u\backslash \left\{\!K_v^{(U)}\!\right\}}}\hspace{-0.5cm}\sqrt{\!Qr_{{v}\!,j}^{-\alpha}}\mathbf{a}_{{v}\!,k}^H\mathbf{h}_{{v}\!,j}x_{j}\!+\mathbf{a}_{{v}\!,k}^H\mathbf{n}_{{v}}.
\end{equation}
The symbol $\mathbf{a}_{v,k}$ is the $k$th column of $\mathbf{A}_{v}$.
According to (\ref{UL_MBS}), the expectation of the received signal power ${\rm S}_{v,k}^{(U)}$ at $v$BS from the $k$th user associated to it is given by
\begin{equation}
  \mathbb{E}\left[{\rm S}_{v,k}^{(U)}\right]=Qr_{v,k}^{-\alpha}.
\end{equation}
\subsection{UL association probabilities}
Based on the derived average received signal power, the UL association probabilities are discussed here.
The $k$th user is associated to MBS in UL when
\begin{eqnarray}
 \mathbb{E}\left[{\rm S}_{M,k}^{(U)}\right]&>&\mathbb{E}\left[{\rm S}_{S,k}^{(U)}\right]\quad\mbox{with DA,}\label{eqn:1}\\
 \mathbb{E}\left[{\rm S}_{M,k}^{(D)}\right]&>&\mathbb{E}\left[{\rm S}_{S,k}^{(D)}\right]\quad\mbox{with CA.\label{eqn:2}}
\end{eqnarray}
Then, note that the distribution of $r_{v,k}$ follows from the null probability of a PPP in an area $\pi r^2$, which means there is no BS in the circle with radius $r$, we have
\begin{equation}\label{27}
  {\rm Pr}(r_{v,k}>r)=e^{-\pi\lambda_vr^2}.
\end{equation}
Based on (\ref{27}), the CDF and the corresponding PDF of $r_{v,k}$ are derived as
\begin{eqnarray}
F_{r_{v,k}}(r)&=&1-e^{-\pi\lambda_vr^2},\quad r\geq0, \label{PPP_CDF}\\
  f_{r_{v,k}}(r)&=&2\pi\lambda_vre^{-\pi\lambda_vr^2},\quad r\geq0.
\end{eqnarray}
For DA scenario, the probability of associating to the nearest $v$BS is represented as $A_v^{(U)}$, while for CA scenario, the probability that each user in the set of $K_{m,n}$ users will associate to the nearest $v$BS is represented by $A_{v,m,n}^{(D)}$. Now, the general expressions for the association probabilities are developed first. Then, several special cases are discussed.
\subsubsection{UL association probabilities with DA}
\

According to (\ref{eqn:1}), the $k$th user is associated to $\mbox{MBS}_m$ in UL with DA when
\begin{equation}
  r_{M,k}^{-\alpha}>r_{S,k}^{-\alpha}.
\end{equation}
Hence, the association probabilities are given by
\begin{eqnarray}
  A_M^{(U)} \!&=&\! {\rm Pr}\left(r_{M,k}<r_{S,k}\right)=\frac{1}{1+\lambda_S/\lambda_M}, \label{33}\\
  A_S^{(U)}\! &=& 1-A_M^{(U)}=\frac{\lambda_S/\lambda_M}{1+\lambda_S/\lambda_M}.\label{34}
\end{eqnarray}

It is quite insightful to see that the association probabilities for every user in the area are identical and independent of the numbers of antennas and the transmit powers of BSs.
In other words, they only depend on the densities of BSs.
Specifically, increasing $\lambda_S$ leads to higher~$A_M^{(U)}$, which means more SBSs deployed in the coverage area of a MBS leads to a higher probability of associating to SBSs. This phenomenon can be explained as follows.
Cell associations of users with DA are based on the maximum received signal power in UL, that is, each user is associated to its nearest BS.
Increasing $\lambda_S$ leads to the decreasing distance between SBSs and users, resulting in more opportunities for users to associate to SBSs.
Besides, it is worth noting that the association probabilities with DA for the multiuser scenario are the same as those in~\cite{6736745} for the single user scenario.
\subsubsection{UL association probabilities with CA}
\

According to (\ref{eqn:2}) and Lemma \ref{lemma01}, the $k$th user is associated to $\mbox{MBS}_m$ in UL with CA when
\begin{eqnarray}\label{33_1}
  \frac{P_{M}\left(L_M\!-\!K_M^{(D)}\!+\!1\right)}{K_M^{(D)}}r_{M,k}^{-\alpha}\!>\!\frac{ P_S\left(L_S\!-\!K_S^{(D)}\!+\!1\right)}{K_S^{(D)}}r_{S,k}^{-\alpha},
\end{eqnarray}
which means the UL association probabilities are equal to the association probabilities in DL.
Note that $K_M^{(D)}$ and~$K_S^{(D)}$ are the total numbers of users associated to $\mbox{MBS}_m$ and $\mbox{SBS}_n$ when the $k$th user is associated to $\mbox{MBS}_m$ and $\mbox{SBS}_n$, respectively. However, there exists two major problems. One is that the $k$th user do not know the accurate values of $K_M^{(D)}$ and~$K_S^{(D)}$, the other is that the $k$th user being added to the system might have an impact on the associations of other users, making it impossible to derive association probabilities for each user. To address the issues, we replace $K_M^{(D)}$ and $K_S^{(D)}$  with the expectations of them, which means the $k$th user knows the average numbers of users associated to $\mbox{MBS}_m$ and $\mbox{SBS}_n$ according to the number and association probabilities of users in each set. The expressions are given by
\begin{eqnarray}
    K_M^{(D)}\hspace{-0.3cm}&=&\hspace{-0.3cm}(K_{m,n}-1)A_{M,m,n}^{(D)}+1+\displaystyle{\sum_{i=1,\atop{i\neq n}}^{N}}K_{m,i}A_{M,m,i}^{(D)},\label{38}\\
     K_S^{(D)}\hspace{-0.3cm}&=&\hspace{-0.3cm}(K_{m,n}-1)A_{S,m,n}^{(D)}+1+\displaystyle{\sum_{j=1,\atop{j\neq m}}^{M}}K_{j,n}A_{S,j,n}^{(D)}.\label{39}
\end{eqnarray}
Then, the association probabilities are derived in the following~theorem.
\begin{theorem}\label{theorem1}
The association probabilities with CA are calculated as\footnote{Since the function $F(A_{M,m,n}^{(D)})=1/\left(1+C_{m,n}^2\lambda_S/\lambda_M\right)-A_{M,m,n}^{(D)}$ is monotonically decreasing in the interval $[0,1]$ with~$F(0)>0$ and $F(1)<~0$, there exits only one root. Therefore, the association probabilities can be derived by solving non-linear equations iteratively with the initial values set as 0 or 1.}
\begin{eqnarray}
  A_{M,m,n}^{(D)} &=&\frac{1}{1+C_{m,n}^2\lambda_S/\lambda_M},\label{36_1}\\
   A_{S,m,n}^{(D)}&=&\frac{C_{m,n}^2\lambda_S/\lambda_M}{1+C_{m,n}^2\lambda_S/\lambda_M},\label{36_2}
\end{eqnarray}
where the constant $C_{m,n}$ is
\begin{equation}\label{37}
    C_{m,n}=\left(\frac{ K_M^{(D)}P_S\left(L_S-K_S^{(D)}+1\right)}{K_S^{(D)}P_M\left(L_M-K_M^{(D)}+1\right)}\right)^{1/\alpha}.
\end{equation}
\end{theorem}
\begin{proof}
From (\ref{33_1}), the probability of associating to MBS can be derived as
\begin{eqnarray}\label{40}
  A_{M,m,n}^{(D)} \!\hspace{-0.4cm}&=&\! \hspace{-0.3cm}{\rm Pr}\!\!\left(\!\!\frac{P_M\!\!\left(\!\!L_M\!\!-\!\!K_M^{(D)}\!\!\!+\!\!1\!\right)}{K_M^{(D)}}\!r_{M\!,k}^{-\alpha}\!\!>\!\!\frac{ P_S\!\!\left(\!L_S\!\!-\!\!K_S^{(D)}\!\!\!+\!\!1\!\right)}{K_S^{(D)}}r_{S\!,k}^{-\alpha}\!\right) \notag\\
  &=&\hspace{-0.3cm}\int_0^{+\infty}\hspace{-0.3cm}\!\!(1\!\!-\!F_{r_{S\!,k}}\!(C_{m,n}r_{M\!,k}))f_{r_{M\!,k}}\!(r_{M\!,k}){\rm d}r_{M\!,k}=\frac{1}{1+C_{m,n}^2\lambda_S/\lambda_M},
\end{eqnarray}
Then, the probability that the $k$th user will be associated to SBS in DL is obtained as
  \begin{equation}
    A_{S,m,n}^{(D)}=1-A_{M,m,n}^{(D)}=\frac{C_{m,n}^2\lambda_S/\lambda_M}{1+C_{m,n}^2\lambda_S/\lambda_M}.
  \end{equation}
\end{proof}

From (\ref{33}) and (\ref{36_1}), it can be noted that the difference between the expressions for association probabilities with DA and CA is the coefficient $C_{m,n}^2$ in (\ref{36_1}). Therefore, besides the densities of BSs, parameters in $C_{m,n}$, shown in (\ref{37}), including the transmit powers as well as the numbers of antennas and users associated to BSs, also have an impact on the association probabilities with CA.

Note that the increase of density of SBSs, i.e., $\lambda_S$, leads to the growth of $\lambda_S/\lambda_M$, which means more SBSs are deployed in the coverage area of a MBS, resulting in higher received powers of users associated to SBSs. Hence, for CA scenario where the association is based on the maximum
DL received power, users are less likely to associate to MBSs, in other words, association probability $A_{M,m,n}^{(D)}$ in (\ref{40}) decreases.
Moreover, decreasing probability of associating to MBS leads to the decreasing $K_M^{(D)}$ and the raising $K_S^{(D)}$. As a result, the coefficient $C_{m,n}$ decreases and $A_{M,m,n}^{(D)}$ grows to some extent. Therefore, the negative feedback enable $A_{M,m,n}^{(D)}$ to converge to a stable value which decreases with the increase of $\lambda_S$.
Similarly, the growth of $P_M$ results in the decrease of $C_{m,n}$ and the increase of $A_{M,m,n}^{(D)}$  which converges to a stable value due to the negative feedback of  $C_{m,n}$. In other words, the raising transmit power of MBSs results in higher received powers of users associated to MBSs, hence, users are more likely to associate to~MBSs.

Next, consider a special case when each BS serves at most one user. Then, we obtain several interesting insights as following corollaries.
\begin{corollary}\label{corollary1}
Since the association probabilities of each user with CA are identical when each BS serves at most one user, the symbols $A_{M,m,n}^{(D)}$ and $A_{S,m,n}^{(D)}$ can be replaced by $A_{M}^{(D)}$ and $A_{S}^{(D)}$, respectively. Then, the closed-form association probabilities are given by
\begin{eqnarray}
  A_{M}^{(D)} &=&\frac{1}{1+\left(\frac{P_SL_S}{P_ML_M}\right)^{\frac{2}{\alpha}}\frac{\lambda_S}{\lambda_M}},\label{41}\\
   A_{S}^{(D)}&=&\frac{\left(\frac{P_SL_S}{P_ML_M}\right)^{\frac{2}{\alpha}}\frac{\lambda_S}{\lambda_M}}{1+\left(\frac{P_SL_S}{P_ML_M}\right)^{\frac{2}{\alpha}}\frac{\lambda_S}{\lambda_M}}.\label{42}
\end{eqnarray}
\end{corollary}
\begin{proof}
If each BS serves at most one user, we have~$K_{m,i}\!=\!0$ where $i$ varies from 1 to~$N$ with $i\neq n$ as well as $K_{j,n}\!=\!0$ where $j$ varies from 1 to $M$ with $j\neq m$. Moreover, note that~$K_{m,n}\!=\!1$. Hence, it is obtained from~(\ref{38}) and (\ref{39}) that $K_M^{(D)}=1$ and $K_S^{(D)}=1$. Then, according to~(\ref{37}), we have $C_{m,n}=\left(\frac{P_SL_S}{P_ML_M}\right)^{\frac{1}{\alpha}}$. By substituting $C_{m,n}$ into (\ref{36_1}) and (\ref{36_2}), the desired results is derived.
\end{proof}
\begin{remark}
Besides the densities of BSs, the association probabilities also depend on the products of the transmit powers and antennas of BSs.
The raising $P_ML_M$, i.e., the increasing transmit power and/or antennas of MBSs, results in the decreasing $C_{m,n}$.
Hence, the association probability~$A_{M}^{(D)}$ increases and $A_{S}^{(D)}$ decreases.
This result reveals that, different from the single antenna scenario, \textbf{array gain} is achieved by using multiple antennas at BSs.
Therefore, more antennas could be equipped in BSs to linearly increase the effective received signal powers at~users.
\end{remark}
Moreover, according to Corollary \ref{corollary1}, note that when the number of antennas equipped in each MBS is the same as that equipped in each SBS, including the case of $L_M=L_S=1$, we have~$\frac{L_S}{L_M}=1$, plugging which into (\ref{41}) and (\ref{42}) derives the association probabilities that are the same as those in \cite{6736745} for the single user scenario.
\begin{corollary}\label{corollary3}
When $P_SL_S\!\!=\!\!P_ML_M$, especially when the transmit power and the number of antennas of MBSs are the same as those of SBSs, the two-tier heterogeneous network becomes a homogeneous network. Then, the association probabilities with CA in (\ref{41}) and (\ref{42}) are identical to those with DA in~(\ref{33}) and (\ref{34}).
\end{corollary}
Corollary \ref{corollary3} implies that when $P_SL_S\!=\!P_ML_M$ in the single user scenario, there exists no difference between MBSs and SBSs from the users' viewpoint.
Hence, the two-tier heterogeneous network becomes a homogeneous network where the cell association policies for CA and DA are identical.
The reason is that since $P_SL_S\!=\!P_ML_M$, the average received signal power in DL only depends on the distance between one user and its nearest $v$BS, thus, the association strategy of DL is the same as that of UL.
Therefore, the UL association probabilities with DA which depend on the UL received power are the same as those with CA which depend on the DL received power. Hence, DA is not necessary in homogeneous~networks.

With the derived association probabilities, the average spectral efficiencies with DA and CA will be analyzed in the following section.
\section{Uplink Spectral Efficiency Analysis}\label{section4}
In this section, the average UL spectral efficiencies with DA and CA are analyzed based on the UL SINR, where we first analyze the average spectral efficiency for each association case and then derive the final results averaged over all association cases.
Due to the difficulty of obtaining the exact expressions for the average spectral efficiency, we derive the lower bounds on them instead.
\subsection{UL Spectral efficiency with DA}
Here, spectral efficiency with DA is considered.
Since the spectral efficiency is given as
\begin{equation}\label{S_eff}
  {\rm SE} \triangleq \mathbb{E}\left[{\rm ln(1+SINR)}\right],
\end{equation}
the UL SINR is required first.
According to (\ref{UL_MBS}), the UL SINR at $v$BS with DA is obtained as
\begin{equation}\label{SINR_M^{(U)}}
  {\rm SINR}=\frac{Qr_{v,k}^{-\alpha}}{I_v^{(U)}+\left\lVert\mathbf{a}_{v,k}\right\rVert^2\sigma^2}.
\end{equation}
Since ZF detector is used, the interference from other users associated to $v$BS is eliminated. Then, the cumulative interference from the users associated to other BSs is given as
\begin{equation}\label{I_r}
  I_v^{(U)}=\sum_{{j\in\Phi_u\backslash \left\{\!K_v^{(U)}\!\right\}}}Qr_{v,j}^{-\alpha}\left\lvert\mathbf{a}_{v,k}^H\mathbf{h}_{v,j}\right\rvert^2.
\end{equation}
Hence, substituting (\ref{SINR_M^{(U)}}) into (\ref{S_eff}) obtains ${\rm SE}_{v,m,n}^{({\rm DA})}$, which is the spectral efficiency of the $k$th user associated to $v$BS,
\begin{eqnarray}\label{tauDA}
  {\rm SE}_{v,m,n}^{({\rm DA})}\!\! =\!\!\mathbb{E}_{{I_v^{(U)}\!,\mathbf{a}_{v\!,k},r_{v\!,k}}}\!\!\!\left[{\rm ln}\!\left(\!1\!\!+\!\frac{Qr_{v\!,k}^{-\alpha}}{I_v^{(U)}\!\!+\!\!\left\lVert\mathbf{a}_{v\!,k}\right\rVert^2\!\!\sigma^2}\!\!\right)\!\right]\!\!.
\end{eqnarray}
Note that the exact expression for ${\rm SE}_{v,m,n}^{({\rm DA})}$ should be derived by computing the expectation with respect to $I_v^{(U)}$,~$\mathbf{a}_{v\!,k}$ and~$r_{v,k}$.
Thus, a triple integral is required, which is extremely complicated. Moreover, the variables are in different parts of the fraction in the logarithmic function, leading to the difficulty in obtaining an exact expression for ${\rm SE}_{v,m,n}^{({\rm DA})}$.
Therefore, a lower bound on~${\rm SE}_{v,m,n}^{({\rm DA})}$ is obtained to replace the exact expression. Since the function~$\phi(x)\!\!=~\!{\rm ln}(1+\frac{1}{x})$ is concave, by using Jensen's inequality, the lower bound on~${\rm SE}_{v,m,n}^{({\rm DA})} $ can be derived as follows,
\begin{equation}\label{inequ}
  {\rm SE}_{v\!,m,n}^{({\rm DA})} \geq\mathbb{E}_{r_{v\!,k}}\left[{\rm ln}\left(1+\frac{Qr_{v,k}^{-\alpha}}{\gamma_{{v}}^{({\rm DA})}}\right)\right],
\end{equation}
where $\gamma_{{v}}^{({\rm DA})}$ is the expectation of interference plus noise, i.e.,
\begin{equation}
  \gamma_{{v}}^{({\rm DA})}=\mathbb{E}_{\mathbf{a}_{v\!,k},r_{v\!,k}}\left[I_v^{(U)}\!\!+\!\!\left\lVert\mathbf{a}_{v\!,k}\right\rVert^2\!\!\sigma^2\right].
\end{equation}
The expression for $\gamma_{{v}}^{({\rm DA})}$ is given in the following lemma.
\begin{lemma}\label{lemma1}
The expression for $\gamma_{{v}}^{({\rm DA})}$ is developed as
\begin{eqnarray}\label{INR51}
  \gamma_{{v}}^{({\rm DA})} \hspace{-0.3cm}  &\triangleq& \hspace{-0.3cm} \frac{Q\overline{r}_{{v}}\!+\!\kappa_v^{(U)}Q\overline{r}_{{v},1}^{({\rm DA})}\!+\!\sigma^2}{L_v-K_v^{(U)}},
\end{eqnarray}
where $\kappa_v^{(U)}=\sum_{n=1}^N K_{S,m,n}^{(U)}$ for $v=M$ and $\kappa_v^{(U)}=K_{M,m,n}^{(U)}$ for $v=S$.
The term $Q\overline{r}_{{v}}$ is the sum of interference from users whose nearest $v$BS is not $\mbox{MBS}_m$ ($\mbox{SBS}_n$) with
\begin{eqnarray}\label{47}
  \overline{r}_{{v}}\triangleq\int_{r_p}^{R}\hspace{-0.2cm}r^{-\!\alpha} 2\pi\lambda_ur(e^{-\lambda_v\pi r_p^2}\!-\!e^{-\lambda_v\pi r^2})e^{-\lambda_w\pi r_p^2}{\rm d}r,
\end{eqnarray}
where $w\in \{M,S\}$ and $w\neq v$.
The term $Q\overline{r}_{{v},1}^{({\rm DA})}$ represents the interference from one user whose nearest $v$BS is $\mbox{MBS}_m$~($\mbox{SBS}_n$) but are associated to $\mbox{SBS}$ ($\mbox{MBS}$) with
\begin{equation}\label{48}
   \overline{r}_{{v},1}^{({\rm DA})}\!\!\triangleq\!\!\frac{2\pi\lambda_v(\lambda_M\!\!+\!\!\lambda_S)}{\lambda_w}\!\!\!\int_{r_p}^{R}\!\!r^{\!-\!\alpha\!+\!1}e^{\!-\!\lambda_v\pi r^2}\!\!\!\left(e^{\!-\!\lambda_w\pi r_p^2}\!-\!e^{\!-\!\lambda_w\pi r^2}\!\right)\!\!{\rm d}r.
\end{equation}
 \end{lemma}
\begin{proof}
The proof is given in Appendix \ref{apx:A}.
\end{proof}
Note that as shown in (\ref{47}) and (\ref{48}), a circular area with radius $R$ is considered, the origin of which is $v$BS.
The users in this area are considered as interference while interference from users out of the area can be ignored.
Moreover, a protective area is considered, in which there are no users. The radius of the protective area is denoted as $r_p$.

Finally, combining (\ref{inequ})-(\ref{48}), the lower bound on the spectral efficiency of $K_v^{(U)}$ users associated to $v$BS is obtained and shown in the following lemma.
\begin{lemma}\label{lemma2}
The spectral efficiency of users associated to $v$BS, i.e., ${\rm SE}_{v,m,n}^{({\rm DA})}$, is lower bounded as
\begin{equation}\label{51M}
{\rm SE}_{v,m,n}^{({\rm DA})}\!\! \geq\!\!  {{\rm ln}\!\left(\!1\!+\!\frac{Q}{\gamma_{{v}}^{({\rm DA})}r_p^{\alpha}}\!\right)}\! {\rm exp}(\!-\pi(\!\lambda_M\!+\!\lambda_S\!)r_p^2)-\!\!\!\int_0^{{\rm ln}\left(\!1+\frac{Q}{\gamma_{{v}}^{({\rm DA})}r_p^{\alpha}}\!\right)}\!\!{\rm exp}\!\!\left(\!\!\!-(\!\lambda_M\!\!+\!\!\lambda_S)\pi\!\!\left(\!\!\frac{Q}{(e^t\!-\!1)\gamma_{{v}}^{{({\rm DA})}}\!}\!\!\right)^\frac{2}{\alpha}\!\right)\!\!{\rm d}t,
\end{equation}
where $\gamma_{{v}}^{({\rm DA})}$ is as defined in Lemma \ref{lemma1}.
\end{lemma}
\begin{proof}
The proof is given in Appendix \ref{apx:B}.
\end{proof}

Then, the average spectral efficiency for a certain association case, i.e.,~${\rm SE}_{m,n}^{({\rm DA})}$, is obtained by averaging the spectral efficiency derived in Lemma \ref{lemma2} over the $K_M$ users whose nearest MBS is $\mbox{MBS}_m$.
Finally, by combining ${\rm SE}_{m,n}^{({\rm DA})}$ with the association probabilities, we derive the average spectral efficiency of the system with DA, represented by ${\rm SE}^{({\rm DA})}$ and given in the following~theorem.

\begin{theorem}\label{theorem2}
The expression for ${\rm SE}^{({\rm DA})}$ is given in (\ref{LB^DA}),
\begin{equation}\label{LB^DA}
{\rm SE}^{({\rm DA})}=\displaystyle{\sum_{n=1}^N}\displaystyle{\sum_{K_{M,m,n}^{(U)}=0}^{K_{m,n}}} \left\{\displaystyle{\prod_{n=1}^N}\left[\binom{K_{M,m,n}^{(U)}}{K_{m,n}}(A_M^{(U)})^{K_{M,m,n}^{(U)}}(A_S^{(U)})^{K_{S,m,n}^{(U)}}\right]{\rm SE}_{m,n}^{({\rm DA})}\right\},
\end{equation}
where~${\rm SE}_{m,n}^{({\rm DA})}$ is calculated as
\begin{eqnarray}\label{57}
  {\rm SE}^{({\rm DA})}_{m,n}\!=\!\!\frac{K_M^{(U)}{\rm SE}^{({\rm DA})}_{M,m,n}\!+\!\displaystyle{\sum_{n=1}^N}K_{S,m,n}^{(U)}{\rm SE}^{({\rm DA})}_{S,m,n}}{K_M}.
\end{eqnarray}
\end{theorem}
\begin{proof}
Since (\ref{57}) is obtained for each certain $K_M^{(U)}$ and $K_{S,m,n}^{(U)}$ where $1\le n\le N$, we multiply ${\rm SE}^{({\rm DA})}_{m,n}$ with~$\displaystyle{\prod_{n=1}^N}\left[\binom{K_{M,m,n}^{(U)}}{K_{m,n}}(A_M^{(U)})^{K_{M,m,n}^{(U)}}(A_S^{(U)})^{K_{S,m,n}^{(U)}}\right]$, which is the corresponding probability, for each possible association case. Finally, the sum of all the results is derived as the final expression for the average spectral efficiency with DA.
\end{proof}
It can be seen from (\ref{INR51}) that the increasing numbers of antennas lead to lower interference and hence higher spectral efficiency, which means significant benefits are brought by multiple antennas.
Hence, BSs equipped with more antennas could be deployed to achieve enhanced spectral efficiency.
Besides, when $\lambda_u$ increases, i.e., more users demand to be served in the area, the interference at each BS becomes severer according to Lemma \ref{lemma1}, thus the spectral efficiency decreases.
Furthermore, note that $P_M$ and $P_S$ have no impact on the spectral efficiency of users with DA for the following reasons.
First, the UL spectral efficiency depends on the UL SINR, which is invariant of the transmit powers of BSs in DL.
Second, the UL association probabilities with DA only depend on the densities of BSs.
Again, the result implies that DA allows UL and DL transmissions to be totally independent and users are enabled to associate to the optimal BSs in both UL and DL. Hence, compared with CA scenario, the UL performance is improved via DA.

Here, consider a special case of $\alpha=4$, then the lower bound in Theorem \ref{theorem2} can be rewritten with exponential integrals, i.e.,~$Ei(\cdot)$, given in Corollary \ref{corollary3+}.
\begin{corollary}\label{corollary3+}
When $\alpha=4$, the lower bound can be simplified by substituting ${\rm SE}_{v\!,m,n}^{({\rm DA})}$ given in~(\ref{51}) into (\ref{LB^DA}) and (\ref{57}),
\begin{eqnarray}\label{51}
 \hspace{-0.8cm}&&\hspace{-0.8cm} {\rm SE}_{v,m,n}^{({\rm DA})} \geq  {{\rm ln}\left(1+\frac{Q}{\gamma_{{v}}^{({\rm DA})}r_p^4}\right)} {\rm exp}(-\pi(\lambda_M+\lambda_S)r_p^2)-\!2\left({\rm Re}\{e^{ib_v}\!Ei\!\left(\!-a\!-\!ib_v\!\right)\}\!-\!Ei(-a)\right)\!\!,
\end{eqnarray}
where $a$ and $b_v$ are defined as
\begin{eqnarray}
  a&=&(\lambda_M+\lambda_S)\pi r_p^2,\\
  b_v&=&(\lambda_M+\lambda_S)\pi\left(\frac{Q}{\gamma_{{v}}^{({\rm DA})}}\right)^{\frac{1}{2}}.
\end{eqnarray}
\end{corollary}

Then, consider a special case when each BS serves at most one user, then we have $K_{m,n}\!=\!1$ as well as $K_{m,i}\!=\!0$ and~$K_{j,n}\!=\!0$ where $i$ varies from 1 to~$N$ with $i\neq n$ and~$j$ varies from~1 to~$M$ with $j\neq m$. Hence, the lower bound on the spectral efficiency is obtained in the following~corollary.
\begin{corollary}\label{corollary4}
When each BS serves at most one user, the spectral efficiency of the system with DA is developed as
\begin{equation}\label{62}
{\rm SE}^{({\rm DA})}= A_M^{(U)}{\rm SE}_{M\!,m,n}^{({\rm DA})}+A_S^{(U)}{\rm SE}_{S,m,n}^{({\rm DA})}.
\end{equation}
The lower bounds on ${\rm SE}_{v\!,m,n}^{({\rm DA})}$  is as derived in Lemma \ref{lemma2} while~$\gamma_{{v}}^{({\rm DA})}$ in ${\rm SE}_{v\!,m,n}^{({\rm DA})}$ is given by
\begin{equation}\label{61_1}
  \gamma_{{v}}^{({\rm DA})} \!\!\triangleq\! \frac{1}{L_v\!-\!1}\left(\!\frac{2Q\pi(\lambda_M\!+\!\lambda_S)\left(R^{2-\alpha}\!-\!r_p^{2-\alpha}\right)}{2-\alpha}\!+\!\sigma^2\!\right)\!\!,
\end{equation}
where $\lambda_M+\lambda_S$ is the density of interfering users.
\end{corollary}
Corollary \ref{corollary4} describes the scenario where each BS serves at most one user in each resource block, which is discussed in previous works, e.g., \cite{7003998,7247166}. However, different from previous works, multiple antennas are considered and ZF detector is utilized here.
Therefore, a comparison between the expectation of the interference plus noise given in \cite{7003998} and that given in Corollary~\ref{corollary4} can be made.

According to a signal model described in \cite{7003998}, we compute the expectation of interference plus noise at $v$BS as follows, where an interference area with radius $R$ and a protective area with radius $r_p$ are still assumed,
\begin{eqnarray}\label{63}
  \mathbb{E}\left[{{I_v^{(U)}}}+\sigma^2\right]=\frac{2Q\pi(\lambda_M\!+\!\lambda_S)(R^{2-\alpha}\!-\!r_p^{2-\alpha})}{2-\alpha}+\sigma^2.
\end{eqnarray}
Based on (\ref{61_1}) and (\ref{63}), it can be seen that by using multiple receive antennas and the ZF detector, the interference plus noise at a $v$BS could be reduced by $\frac{1}{L_v-1}$, which shows significant interference suppression brought by multiple antennas.
\subsection{UL Spectral efficiency with CA}
In this subsection, a lower bound on the average spectral efficiency with CA is explored to make a comparison with DA. First, the expectation of interference plus noise at $\mbox{MBS}_m$, i.e.,~$\gamma_{{M}}^{({\rm CA})}$, is presented in the following lemma.
\begin{lemma}\label{lemma4}
The expression for $\gamma_{{M}}^{({\rm CA})}$ is developed as
\begin{equation}\label{61}
  \gamma_{{M}}^{({\rm CA})}\triangleq\frac{Q\overline{r}_{{M}}+\sum_{n=1}^NK_{S,m,n}^{(D)}Q\overline{r}_{{M},1}^{({\rm CA})}+\sigma^2}{L_M-K_M^{(D)}},
\end{equation}
where $Q\overline{r}_{{M},1}^{({\rm CA})}$ is the interference from each user associated to SBSs while its nearest MBS is $\mbox{MBS}_m$, given by
\begin{equation}\label{64}
  \overline{r}_{{M}\!,1}^{({\rm CA}\!)}\!\!\triangleq\!\!\frac{2\pi\widetilde{\lambda}\lambda_M}{\lambda_SC_{m\!,n}^2}\!\!\int_{\frac{r_p}{C_{m\!,n}}}^R \!\!\!\!\!r^{\!-\!\alpha\!+\!1}e^{\!-\!\lambda_M\pi r^2}\!\left(\!e^{\!-\!\lambda_S\pi r_p^2}\!-\!e^{\!-\!\lambda_S\pi C_{m\!,n}^2r^2}\!\right)\!\!{\rm d}r,
\end{equation}
in which $\widetilde{\lambda}$ is defined as
\begin{equation}
\widetilde{\lambda}=\lambda_M\!+\!\lambda_SC_{m,n}^2.
\end{equation}
\end{lemma}
\begin{proof}
Since the first part of the interference, i.e., $Q\overline{r}_{{M}}$, is the same as that in $\gamma_{{M}}^{({\rm DA})}$, only the expression for $ \overline{r}_{{M}\!,1}^{({\rm CA}\!)}$ is required here. First, we obtain the CDF of the distance between $\mbox{MBS}_m$ and users associated to SBSs while their nearest MBS is~$\mbox{MBS}_m$,
\begin{equation}\label{92}
 F_{M}^{(\!{\rm CA}\!)}\!(\!r\!)\!=\!\!\frac{{\rm Pr}\!\left(\!r_p\!\!<\!r_{M\!,k}\!\!<\!r\!,\!r_p\!\!<\!r_{S,k}\!\!<\!\!C_{m,n}r_{M,k}\!\right)}{A_{S,m,n}^{(D)}} \!\!=\!\!\frac{2\pi\widetilde{\lambda}\lambda_M}{\lambda_SC_{m\!,n}^2}\!\!\int_{\!\frac{r_p}{C_{m\!,n}}}^r \!\!\!\!\!r_{M\!,k}e^{\!-\!\lambda_{\!M}\!\pi r_{\!M\!,k}^2}\!\!\left(\!e^{\!-\!\lambda_S\!\pi r_p^2}\!\!-\!e^{\!-\!\lambda_S\!\pi C_{m\!,n}^2\!r_{M\!,k}^2}\!\!\right)\!\!{\rm d}r_{M\!,k}.
\end{equation}
Then, we obtain (\ref{64}) using the corresponding PDF.
\end{proof}
With the derived expectation of interference plus noise, a lower bound on ${\rm SE}_{M,m,n}^{{\rm CA}}$, which is the spectral efficiency of users associated to $\mbox{MBS}_m$, is obtained in the following lemma.
\begin{lemma}\label{lemma5}
The spectral efficiency of users associated to $\mbox{MBS}_m$ with CA is lower bounded by
\begin{equation}\label{SE_M^CA}
{\rm SE}_{M,m,n}^{{\rm (CA)}}\!\! \geq\!\! L_0e^{\!-\pi\widetilde{\lambda}\left(\!\frac{r_p}{C_{m,n}}\!\right)^2\!}\!\!\!+\!\frac{(\!L_1\!-\!L_0\!)\widetilde{\lambda}}{\lambda_M}
  e^{\!-(\!\lambda_M\!+\!\lambda_S\!)\pi r_p^2}\!-\!\!\!\int_0^{L_0} \!\!e^{\!-\widetilde{\lambda}\!\pi\!\left(\!\!\frac{Q}{(\!e^t\!-\!1)\gamma_{{M}}^{(\!{\rm CA}\!)}}\!\!\right)^{\frac{2}{\alpha}}\!}\!\!\!\!\!{\rm d}t-\!\frac{\widetilde{\lambda}e^{-\!\lambda_S\pi r_p^2}}{\lambda_M}\!\!\!\int_{L_0}^{L_1} \!\! e^{-\!\lambda_M\pi\!\left(\!\!\frac{Q}{(\!e^t\!-\!1\!)\gamma_{{M}}^{(\!{\rm CA}\!)}}\!\!\right)^{\frac{2}{\alpha}}}\!\!\!\!\!\!{\rm d}t,
\end{equation}
where $L_0$ and $L_1$ are defined as
\begin{eqnarray}\label{INR71}
  L_0&=&{\rm ln}\left(1+\frac{Q}{\gamma_{{M}}^{({\rm CA})}(r_p/C_{m,n})^{\alpha}}\right),\label{71_1}\\
  L_1&=&{\rm ln}\left(1+\frac{Q}{\gamma_{{M}}^{({\rm CA})}r_p^{\alpha}}\right).\label{71_2}
\end{eqnarray}
\end{lemma}
\begin{proof}
The proof is given in Appendix \ref{apx:C}.
\end{proof}
By following similar steps, a lower bound on ${\rm SE}^{({\rm CA})}_{S,m,n}$, which is the spectral efficiency of $\mbox{SBS}_n$ with CA, is obtained in Lemma \ref{lemma6}.
\begin{lemma}\label{lemma6}
The lower bound on ${\rm SE}^{({\rm CA})}_{S,m,n}$ is derived as
\begin{eqnarray}\label{68}
  &&\hspace{-1.1cm}{\rm SE}^{({\rm CA})}_{S,m,n} \geq{\rm exp}\!\!\left(\!-\frac{\widetilde{\lambda}\pi r_p^2}{C_{m,n}^2}\!\right)\!{{\rm ln}\!\!\left(\!1\!+\!\frac{Q}{\gamma_{{S}}^{({\rm CA})}r_p^{\alpha}}\!\right)}\!\!-\!\!\!\int_0^{{\rm ln}\left(1\!+\!\frac{Q}{\gamma_{{S}}^{({\rm CA})}r_p^{\alpha}}\right)}{\rm exp}\!\!\left(\!\!-\frac{\widetilde{\lambda}\,\pi}{C_{m,n}^2}\left(\!\!\frac{Q}{(e^t\!-\!1)\gamma_{{S}}^{({\rm CA})}}\!\right)^{\!\frac{2}{\alpha}}\!\right)\!\!{\rm d}t,
\end{eqnarray}
where the symbol $\gamma_{{S}}^{({\rm CA})}$ is
\begin{eqnarray}\label{67}
\gamma_{{S}}^{({\rm CA})}\!\triangleq\!\left(Q\overline{r}_{{S}}\!+\! QK_{M,m,n}^{(D)}\overline{r}_{{S},1}^{({\rm CA})}\!+\!\sigma^2\right)\frac{1}{L_S-K_S^{(D)}},
\end{eqnarray}
with $\overline{r}_{{S},1}^{({\rm CA})}$ given by
\begin{equation}\label{70_2}
 \! \overline{r}_{{S}\!,1}^{(\!{\rm CA}\!)}\!\!\triangleq\!\!\frac{2\pi\lambda_S\widetilde{\lambda}}{\lambda_M}\!\!\!\int_{r_p}^{R}\!\! r^{\!-\!\alpha\!+\!1}e^{\!-\!\lambda_S\!\pi r^2}\!\left(\!\!e^{\!-\!\lambda_M\!\pi r_p^2}\!-\!e^{\!-\!\lambda_M\!\pi\!\frac{r^2}{C_{m,n}^2}}\!\right)\!\!{\rm d}r.
\end{equation}
\end{lemma}

\begin{proof}
The CDF of the distance from one user with CA to its tagged SBS is given by
\begin{equation}\label{71}
  F_{r_{S,k}}^{(D)}(\!r\!)\!\!=\!\!{\frac{{\rm Pr}\left(r_p\!\!<\!\!r_{S,k}\!\!<\!r,r_{M,k}\!>\!\frac{1}{C_{m,n}}r_{S,k},r_{M,k}\!>\!r_p\!\right)}{A_{S,m,n}^{(D)}}} \!=\!{\rm exp}\!\!\left(\!\!-\frac{\widetilde{\lambda}}{C_{m,n}^2}\pi r_p^2\!\right)\!-\!{\rm exp}\!\!\left(\!\!-\frac{\widetilde{\lambda}}{C_{m,n}^2}\pi r^2\!\right).
\end{equation}
Similar to the derivation of ${\rm SE}^{({\rm CA})}_{M,m,n}$, by using (\ref{71}), the expression in~(\ref{68}) is derived.
Moreover, the PDF of the distance between $\mbox{SBS}_n$ and the users associated to MBSs while their nearest SBS is $\mbox{SBS}_n$ is given as
\begin{eqnarray}\label{70_new}
f_{S}^{(\!{\rm CA}\!)}\!(\!r\!)\!\!=\!\!\frac{{\rm d}\!\!\left(\!{\rm Pr}\!\!\left(\!r_p\!<\!r_{S,k}\!<\!r\!<\!R,\!r_p\!<\!r_{M,k}\!<\!\frac{r_{S,k}}{C_{m,n}}\!\right)\!\!\right)}{A_{M,m,n}^{(D)}{\rm d}r}\! =\!\frac{2\pi\lambda_S\widetilde{\lambda}}{\lambda_M}re^{-\!\lambda_S\pi r^2}\left(\!e^{-\!\lambda_M\pi r_p^2}\!\!-\!e^{\!-\!\lambda_M\pi \!\left(\!\frac{r}{C_{m\!,n}}\right)^2}\right).
\end{eqnarray}
Then, the desired result in (\ref{70_2}) is obtained by using (\ref{70_new}).
\end{proof}
Based on the spectral efficiency of users associated to $\mbox{MBS}_m$ and $\mbox{SBS}_n$, the spectral efficiency averaged over $K_M$ users, denoted as ${\rm SE}^{({\rm CA})}_{m,n}$, is derived. Then, each possible value of ${\rm SE}^{({\rm CA})}_{m,n}$ is multiplied with the association probability given in Theorem \ref{theorem1}. Finally, the sum of them is obtained as the average spectral efficiency of the system for CA scenario.
\begin{theorem}
The average spectral efficiency of the system for CA scenario, denoted by ${\rm SE}^{({\rm CA})}$, is shown in (\ref{72}),
\begin{equation}\label{72}
  {\rm SE}^{({\rm CA})}=\displaystyle{\sum_{n=1}^N}\displaystyle{\sum_{K_{M,m,n}^{(D)}=0}^{K_{m,n}}} \left\{\displaystyle{\prod_{n=1}^N}\left[\binom{K_{M,m,n}^{(D)}}{K_{m,n}}(A_{M,m,n}^{(D)})^{K_{M,m,n}^{(D)}}(A_{S,m,n}^{(D)})^{K_{S,m,n}^{(D)}}\right]{\rm SE}^{({\rm CA})}_{m,n}\right\},
  \end{equation}
 where~${\rm SE}^{({\rm CA})}_{m,n}$ is given as
\begin{equation}\label{70}
   {\rm SE}^{({\rm CA})}_{m,n}=\frac{\displaystyle{\sum_{n=1}^N}\!\!\left\{\!K_{M,m,n}^{(D)}{\rm SE}^{({\rm CA})}_{M,m,n}\!+\!\!K_{S,m,n}^{(D)}{\rm SE}^{({\rm CA})}_{S,m,n}\right\}}{K_M}.
\end{equation}
\end{theorem}
Similar to the spectral efficiency of users with DA, the increasing numbers of antennas lead to lower interference and thus higher spectral efficiency.
Moreover, according to Lemma~\ref{lemma4} and Lemma \ref{lemma6}, the denser users located in a certain area are, the higher the interference is, hence the lower the spectral efficiency of users with CA is.
However, different from the DA scenario, the average spectral efficiency of users with CA also depends on the transmit powers of BSs. The reason is that although the spectral efficiency is calculated using the UL SINR in which the DL transmit powers of BSs are not involved, the association probabilities of users with CA, which are required for the average spectral efficiency, depend on the transmit powers of BSs.

Consider the special case mentioned in Corollary \ref{corollary4}, we derive Corollary \ref{corollary6}.
\begin{corollary}\label{corollary6}
When the products of the transmit powers and antennas of MBSs and SBSs are identical, including the case when $P_M\!=\!P_S$ and $L_M\!=\!L_S$, the lower bound on the spectral efficiency for CA scenario is the same as that for DA scenario. Therefore, DA is not necessary in a homogeneous network.
\end{corollary}
\begin{proof}
Note that if $P_SL_S=P_ML_M$, association probabilities with DA are the same as those with CA according to Corollary \ref{corollary3}.
Meanwhile, when $C_{m,n}=1$, we have
\begin{equation}
  {\rm SE}_{v,m,n}^{{\rm (CA)}}\!\geq \!\!\!\int_0^{{\rm ln}\left(1+\frac{Q}{\gamma_{{v}}^{({\rm CA})}r_p^{\alpha}}\right)} \!\!{\rm exp}\left(\!-\!\pi\!\left(\!\lambda_M\!+\!\lambda_S\!\right)r_p^2\!\right)\!-\!{\rm exp}\!\!\left(\!-\left(\lambda_M\!+\!\lambda_S\right)\pi\left(\frac{Q}{(e^t-1)\gamma_{{v}}^{({\rm CA})}}\right)^{\frac{2}{\alpha}}\!\right)\!\!{\rm d}t.
\end{equation}
Since $\gamma_{{v}}^{({\rm CA})}=\gamma_{{v}}^{({\rm DA})}$ in this situation, the lower bound on~${\rm SE}_{v,m,n}^{{\rm (CA)}}$ is the same as that on ${\rm SE}_{v,m,n}^{{\rm (DA)}}$.
Then, we derive the desired result.
\end{proof}
Corollary \ref{corollary6} coincides with the analysis in Section \ref{section3}-B where we pointed out that, in the single user scenario, there is no difference between DA and CA in a homogeneous network where $P_ML_M=P_SL_S$, since the UL cell association policy for DA scenario is actually the same as that for CA scenario.
\section{Numerical Results}\label{section5}
In this section, the analytical results of association probabilities are validated by simulations. Then, the lower bounds on the spectral efficiency and the simulations are provided.
Finally, a comparison of the spectral efficiency between the DA and CA scenarios is presented based on the lower bounds and the simulations.

Here, a two-tier heterogeneous network is considered and Monte Carlo trials are conducted to obtain association probabilities and the spectral efficiency. The main parameters for the simulations are listed in TABLE \ref{table1}.

\begin{table}
  \centering
  \caption{Simulation Parameters}\label{table1}
  \begin{tabular}{|c|c|}
    \hline
    Parameter & Value\\
    \hline
    $Q$ & 0.1W\\
    \hline
    $P_S$ &0.1W\\
    \hline
    $\sigma^2$ & $10^{-12}$ \\
    \hline
    $\lambda_M$ & $1\times10^{-7}$ \\
    \hline
    $\alpha$ & 4\\
    \hline
    $R$ & 2000 m\\
    \hline
    \end{tabular}
\end{table}
Firstly, the variations of association probabilities with the transmit power of MBSs under different densities of SBSs are shown in Fig. \ref{Prob_P_M}.
It can be observed that with the increase of $P_M$, users with CA are getting more likely to associate with MBSs due to the higher DL received powers at users associated to MBSs, while the association probabilities for DA remains constant.
The reason is that UL cell associations of users with DA only depend on the UL received signal powers at BSs, which are not related to $P_M$.
Furthermore, as shown in Fig. \ref{Prob_P_M}, increasing density of SBSs means more SBSs are deployed in the coverage area of a MBS. Intuitively, the DL received powers at users associated to SBSs increase.
Meanwhile, higher density of SBSs results in shorter distances between users and SBSs, thus leading to higher UL received powers at SBSs.
Hence, users with both CA and DA are less likely to associate with MBSs in UL.
%
\begin{figure}[h]
\begin{minipage}[t]{0.5\linewidth}
\centering
\includegraphics[width=0.9\textwidth]{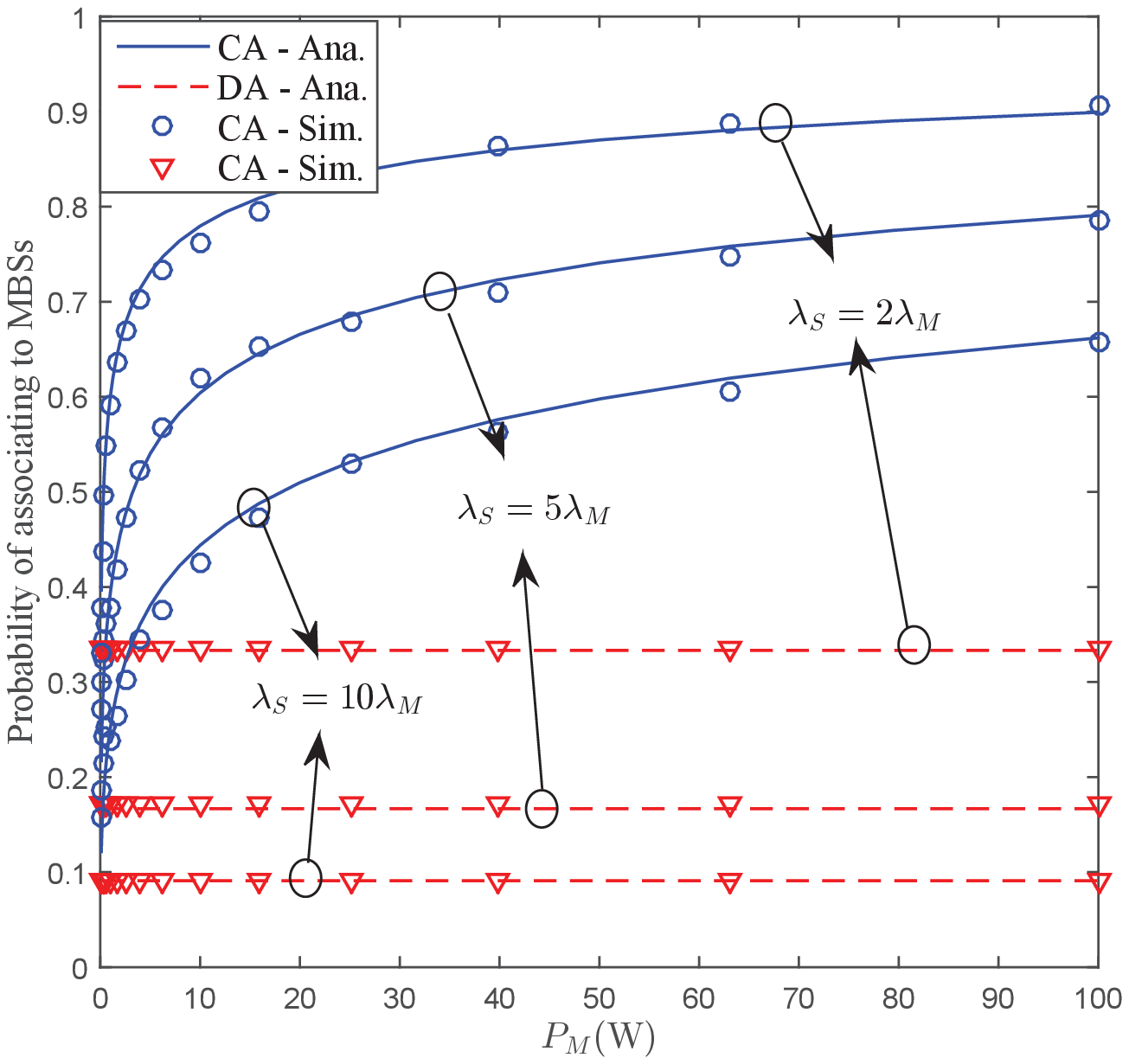}
\caption{Probability of associating to MBSs versus the trans\textrm{-}\protect\\mit power of MBSs for different densities of SBSs ( $\lambda_S$ = \protect\\ $[2\lambda_M,\! 5\lambda_M,\! 10\lambda_M]$) with $\lambda_u\!=\!10\lambda_M$, $L_M\!\!=\!100$ and $L_S\!=\!50$. \label{Prob_P_M}}
\end{minipage}
\hfill
\begin{minipage}[t]{0.5\linewidth}
\centering
\includegraphics[width=0.9\textwidth]{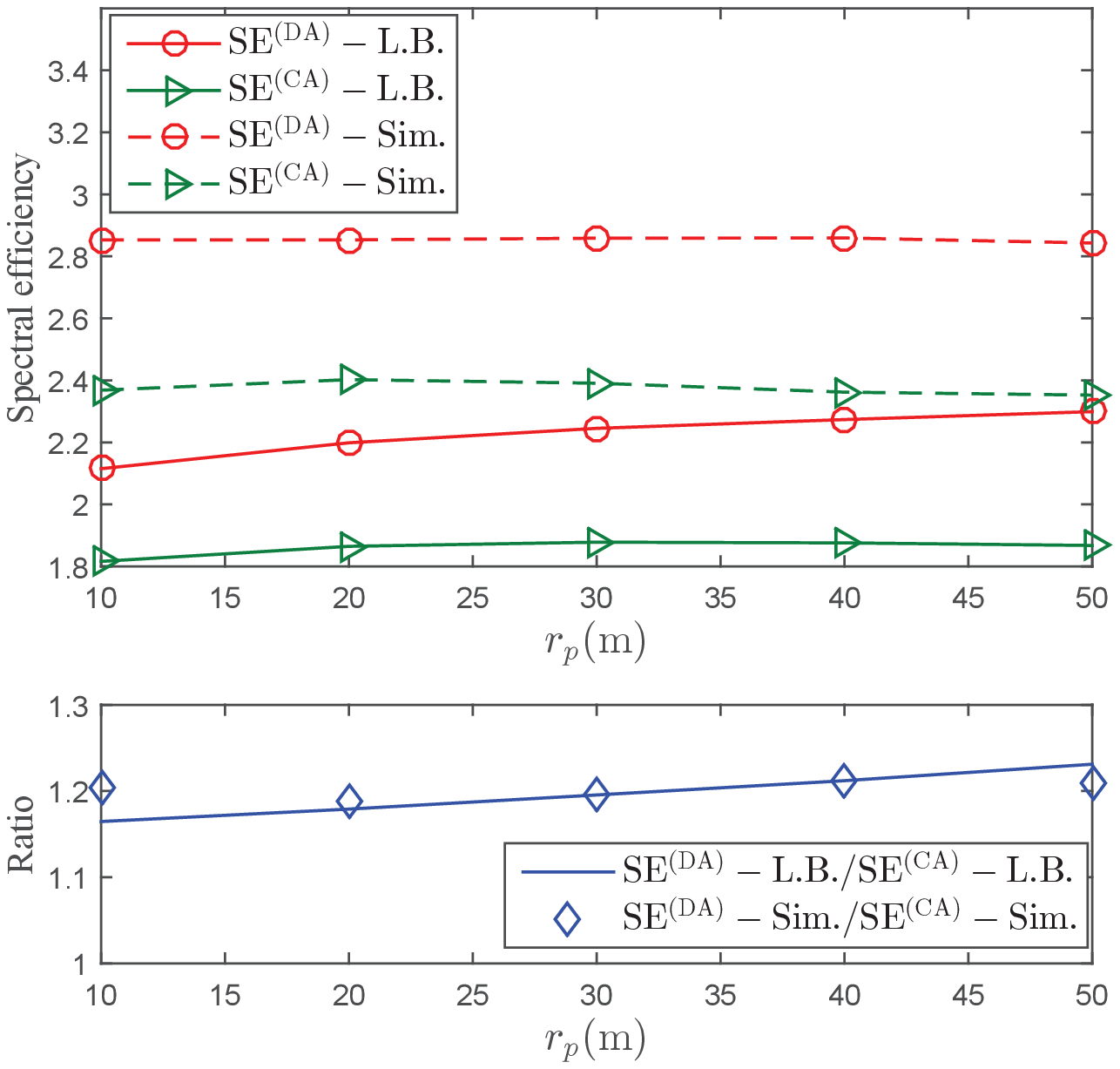}
\caption{Spectral efficiency versus the radius of protection area with $\lambda_S=2\lambda_M$, $\lambda_u=10\lambda_M$, $P_M=40W$,~$L_M=100$~and~$L_S=50$.\label{SE_pro}}
\end{minipage}
\end{figure}

Secondly, the spectral efficiency versus the radius of protection area is illustrated in Fig. \ref{SE_pro}.
As shown, with increasing $r_p$, the spectral efficiencies with DA and CA almost remain unchanged, resulting in the unchanged ratio between~${\rm SE}^{({\rm DA})}$ and ${\rm SE}^{({\rm CA})}$. However, the lower bound on the spectral efficiency with DA increases slightly, leading to an increasing ratio between the lower bounds.
Furthermore, note that the ratio between the lower bounds is quite close to the ratio between the simulations, which means the lower bounds derived in Section~\ref{section4} are able to describe the trends of the spectral efficiencies with DA and CA and the difference between them.
It is also shown in Fig. \ref{SE_pro} that when $r_p=30$m, the ratio between the lower bounds on~${\rm SE}^{({\rm DA})}$ and ${\rm SE}^{({\rm CA})}$ is almost the same as that between the simulations. Therefore, the radius~$r_p$ is set as $30$m for all the rest simulations.
\begin{figure}[h]
\begin{minipage}[t]{0.5\linewidth}
\centering
\includegraphics[width=0.9\textwidth]{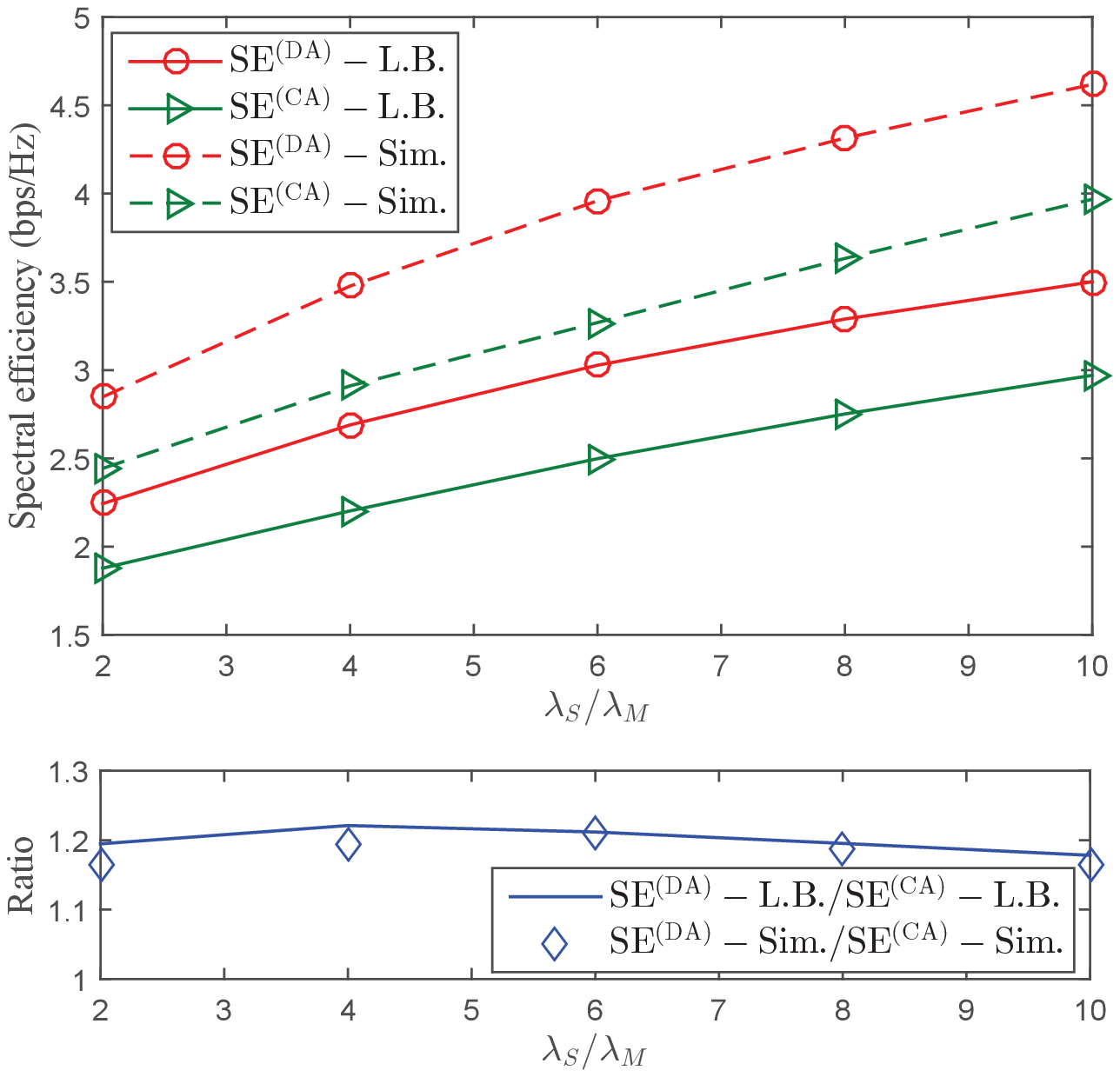}
\caption{Spectral efficiency versus the number of SBS in the \protect\\coverage area of a MBS with $\lambda_u=10\lambda_M$, $P_M=40W$, $L_M\protect\\=100$ and $L_S=50$. \label{SE_lambdaS}}
\end{minipage}
\hfill
\begin{minipage}[t]{0.5\linewidth}
\centering
\includegraphics[width=0.9\textwidth]{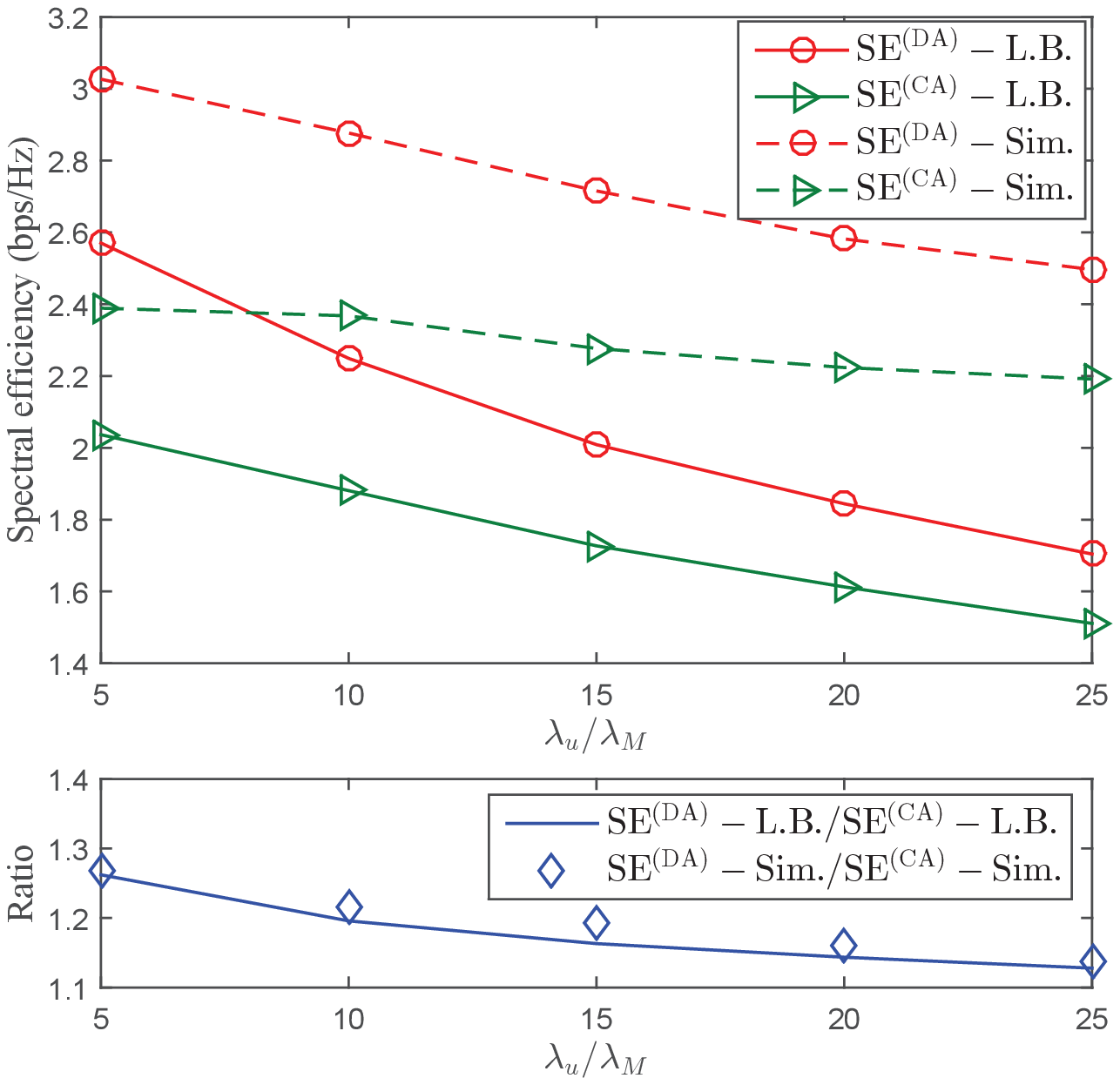}
\caption{Spectral efficiency versus the number of users in the coverage area of a MBS with $\lambda_S=2\lambda_M$, $P_M=40W$, $L_M=100$ and $L_S=50$.\label{SE_lambdad}}
\end{minipage}
\end{figure}

Thirdly, Fig. \ref{SE_lambdaS} and Fig. \ref{SE_lambdad} presents the effect of the number of SBSs and users in the coverage area of a MBS on the spectral efficiency, respectively. Fig.~\ref{SE_lambdaS} shows that with the increase of~$\lambda_S/\lambda_M$, the spectral efficiencies with both DA and CA increase.
This is due to the fact that with the increase of~$\lambda_S/\lambda_M$, more SBSs are deployed in the coverage area of a MBS.
Then, the distance from one user to the nearest SBS becomes smaller, resulting in higher received power at SBSs, which enhances the spectral efficiency. Note that the ratio between the lower bounds with DA and CA as well as the ratio between the simulations are basically unchanged, which means the impact of $\lambda_S/\lambda_M$ on the spectral efficiency with DA is almost the same as that with CA.
Moreover, Fig.~\ref{SE_lambdad} shows that the spectral efficiencies with both DA and CA decrease with the increase of $\lambda_u/\lambda_M$, which is a result of the increasing interference at BSs.
In addition, it can also be noted that the increasing~$\lambda_u/\lambda_M$ leads to the slightly decrease in the ratio between the spectral efficiencies with DA and CA.
The reason is that with the increase of~$\lambda_u/\lambda_M$, the probability of associating to SBS becomes smaller. Hence, the interference at SBSs, which is,~$\gamma_S^{({\rm CA})}$, increases slower than~$\gamma_M^{({\rm CA})}$, resulting in the fact that ${\rm SE}_S^{({\rm CA})}$ decreases slower than~${\rm SE}_M^{({\rm CA})}$.
From~(\ref{70}), it can be seen that since $K_{S,m,n}^{(D)}$ increases faster than $K_{M,m,n}^{(D)}$ and~${\rm SE}_S^{({\rm CA})}$ decreases slower than~${\rm SE}_M^{({\rm CA})}$, the average spectral efficiency with CA decreases slower than that with DA, where~${\rm SE}_M^{({\rm DA})}$ and~${\rm SE}_S^{({\rm DA})}$ decrease at almost the same speed.
\begin{figure}[h]
\begin{minipage}[t]{0.5\linewidth}
\centering
\includegraphics[width=0.9\textwidth]{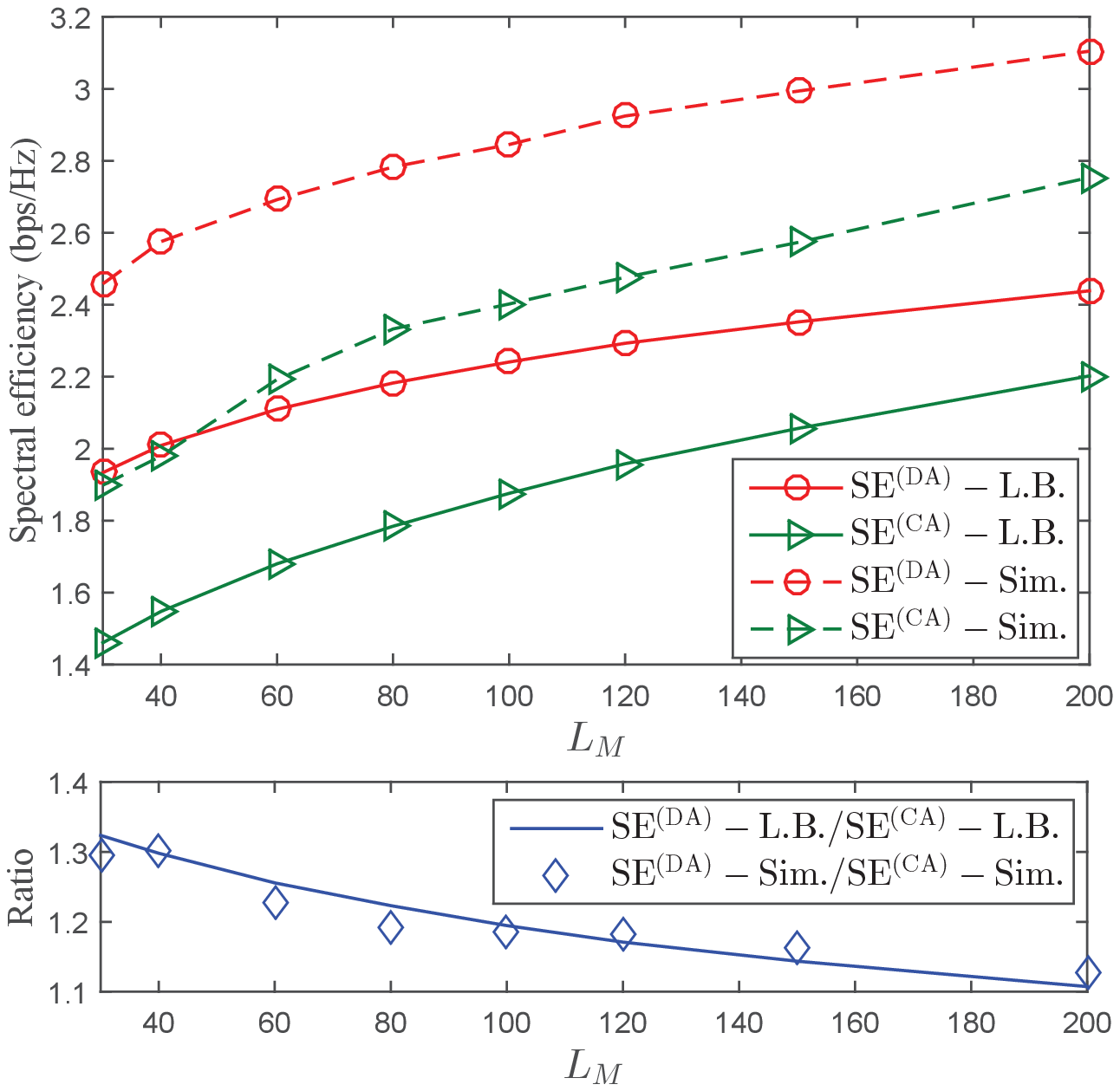}
\caption{Spectral efficiency versus the number of MBS ante\textrm{-}\protect\\nnas with $L_S\!=\!50$, $\lambda_S\!=\!2\lambda_M$, $\lambda_u\!=\!10\lambda_M$ and $P_M\!=\!40W$. \label{SE_L_M}}
\end{minipage}
\hfill
\begin{minipage}[t]{0.5\linewidth}
\centering
\includegraphics[width=0.9\textwidth]{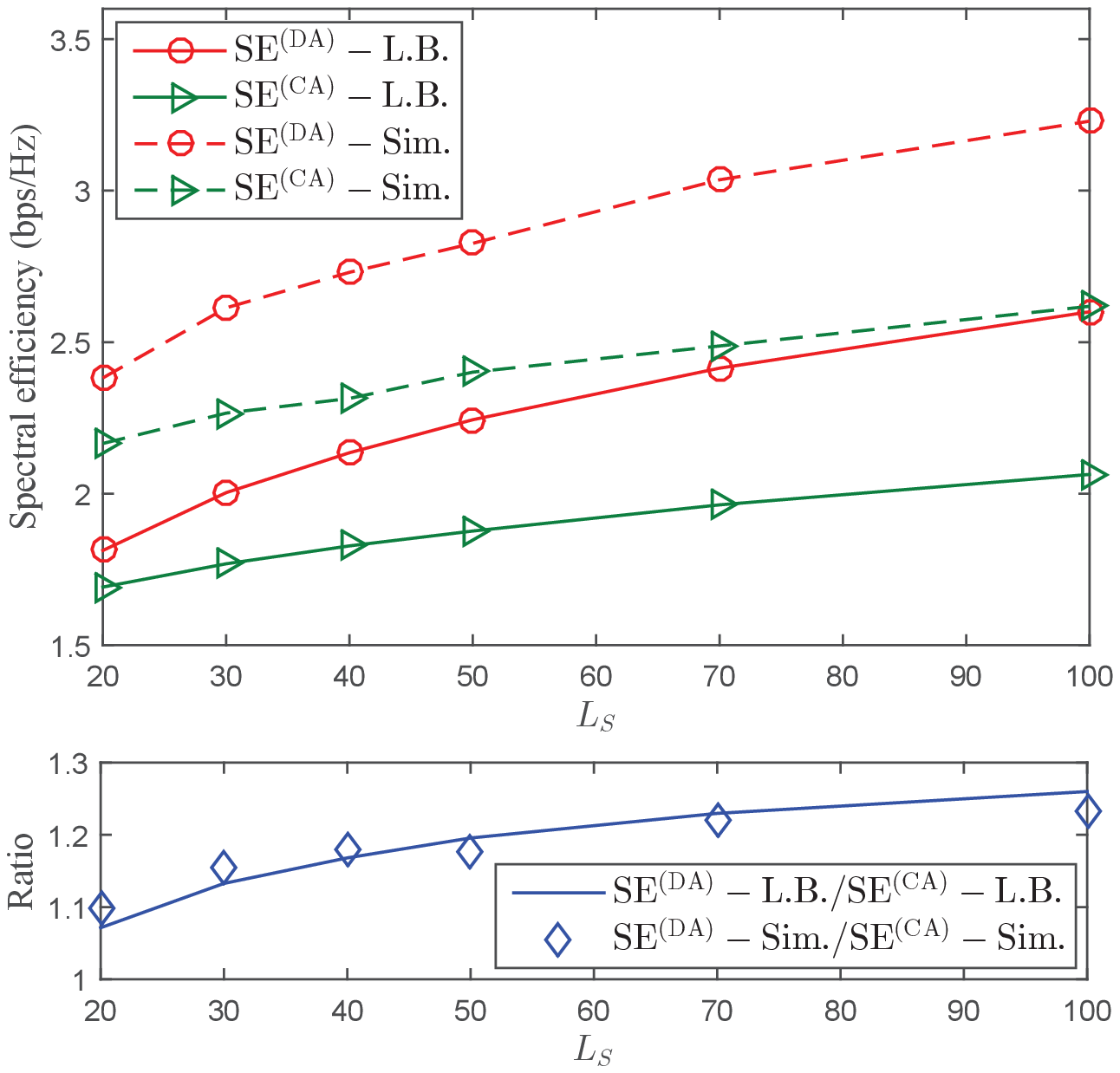}
\caption{Spectral efficiency versus the number of SBS antennas with $L_M=100$, $\lambda_S=2\lambda_M$, $\lambda_u=10\lambda_M$ and $P_M=40W$.\label{SE_L_S}}
\end{minipage}
\end{figure}

Furthermore, Fig. \ref{SE_L_M} presents the effect of $L_M$ on the spectral efficiency, while Fig. \ref{SE_L_S} shows the spectral efficiency when $L_S$ varies from 20 to 100. It can be seen that with the increase of $L_M$ or $L_S$, the spectral efficiencies for both DA and CA scenarios increase.
This is due to the decrease of the expectations of interference plus noise for both DA and CA. Hence, more antennas could be equipped in BSs to enhance the spectral efficiency.
Interestingly, the ratio of spectral efficiency between DA and CA scenarios decreases in Fig. \ref{SE_L_M} and increases in Fig. \ref{SE_L_S}.
It is worth noting that $A_{M,m,n}^{(D)}$ is higher than $A_{S,m,n}^{(D)}$ according to Fig. \ref{Prob_P_M}. More users associating to MBSs results in less interference at MBSs, which implies that ${\rm SE}_M^{({\rm CA})}$ is higher than ${\rm SE}_S^{({\rm CA})}$.
Besides, deploying more MBS antennas results in a higher probability to connect with MBS for CA, contributing to a faster growth in the spectral efficiency in Fig. \ref{SE_L_M}. However, deploying more SBS antennas leads to a higher probability of connecting to SBS.
Thus, the spectral efficiency with CA increases slower than that with DA, as shown in Fig. \ref{SE_L_S}.

Moreover, the variations of the spectral efficiency with the antennas equipped in each SBS are also shown in Fig.~\ref{antennapower}.
However, different from Fig. \ref{SE_L_S}, the sum of antennas in the coverage area of a MBS is considered to be a constant~($[150, 200, 300]$).
Besides, the sum of the transmit power of all the SBSs in the area is also a constant ($0.2W$).
As shown, the increasing $L_S$ leads to the decreasing spectral efficiency.
This result implies that for the limited antennas scenario, instead of BSs with large number of antennas, denser SBSs with less antennas could be deployed to achieve better spectral efficiency.
Furthermore, it can be seen that the increasing number of antennas deployed in a certain area results in the increasing spectral efficiency for each certain value of $L_S$.
This result implies that, in the limited power scenario, instead of SBSs with high transmit power, denser SBSs with lower transmit power~could~be~deployed~to~improve~the~spectral~efficiency.
\begin{figure}[h]
\begin{minipage}[t]{0.5\linewidth}
\centering
\includegraphics[width=0.9\textwidth]{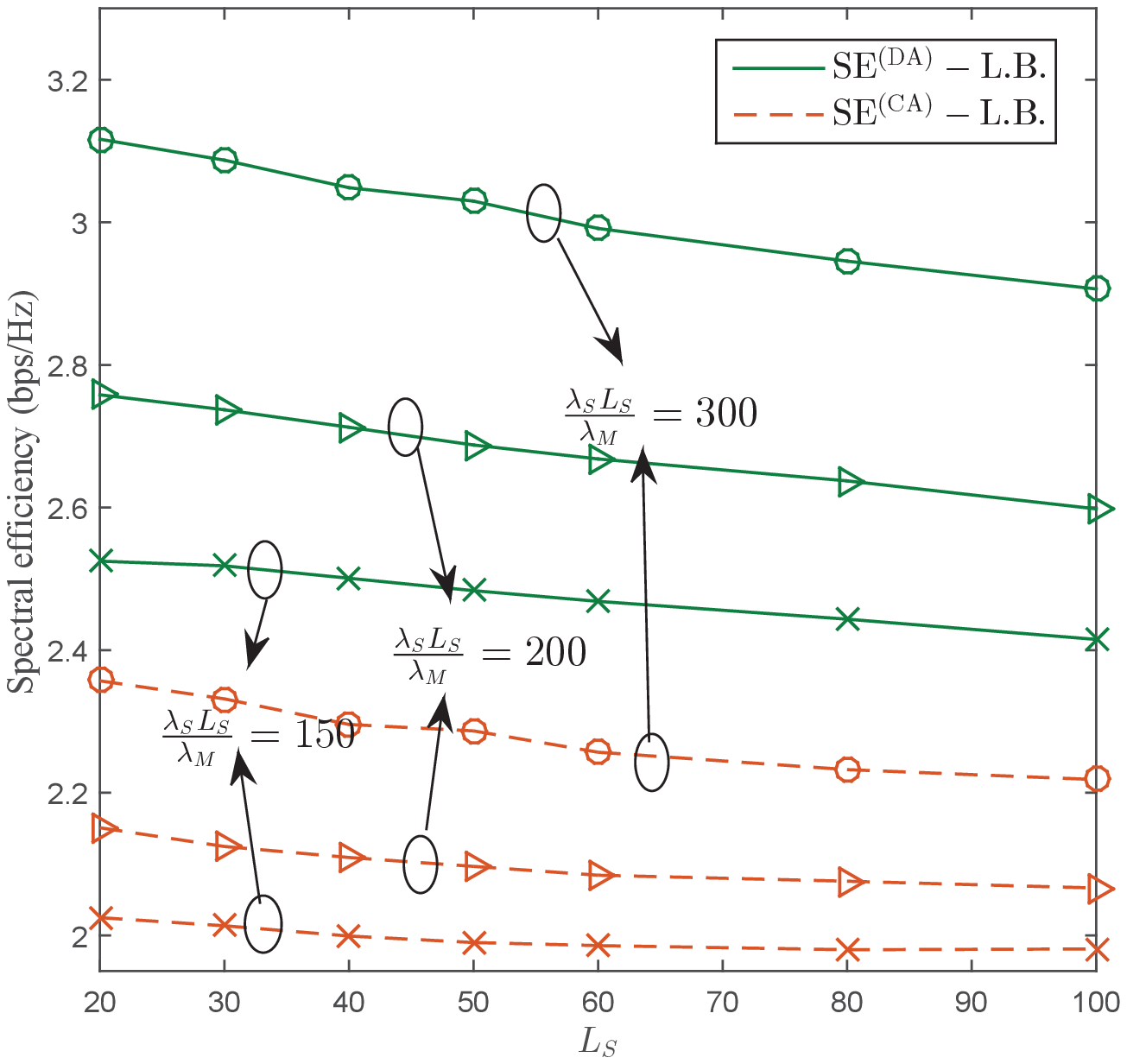}
\caption{Spectral efficiency versus number of antennas each\protect\\ SBS equipped when the sum of SBS antennas in the cover\textrm{-}\protect\\age area of a MBS is $[\,150,\;\;200,\;\;300\,]$, respectively, with\protect\\$\lambda_u=10\lambda_M$, $L_M=100$ and $P_M=40W$. \label{antennapower}}
\end{minipage}
\hfill
\begin{minipage}[t]{0.5\linewidth}
\centering
\includegraphics[width=0.9\textwidth]{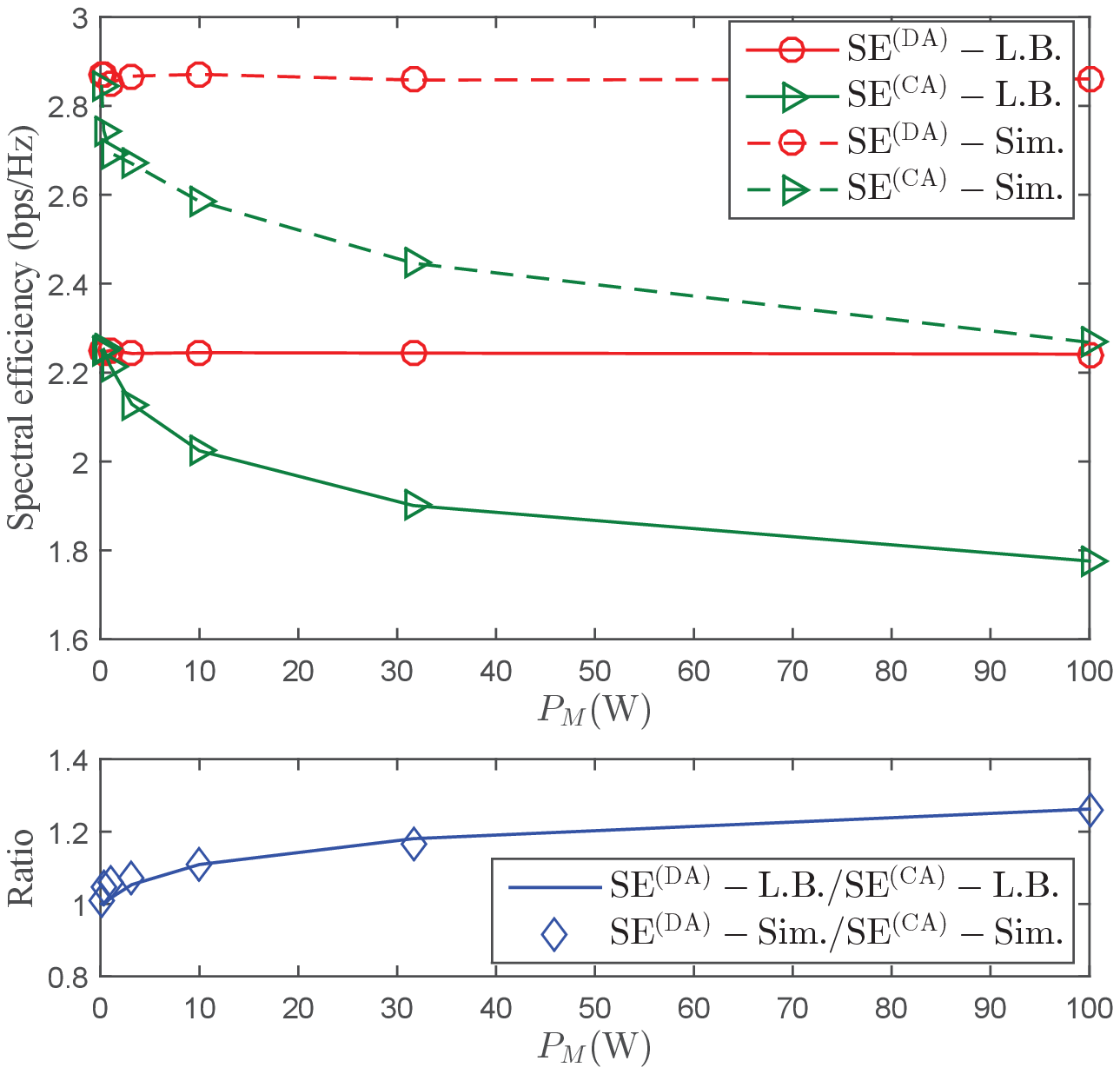}
\caption{Spectral efficiency for different transmit powers of MBSs with $\lambda_S=2\lambda_M$, $\lambda_u=10\lambda_M$, $L_M=100$ and $L_S=50$.\label{SE_P_M}}
\end{minipage}
\end{figure}

Finally, the relationship between the transmit power of MBSs and the spectral efficiency is considered.
Fig. \ref{SE_P_M} shows that with the growth of~$P_{M}$, the spectral efficiency with DA remains constant, since the association probabilities with DA only depend on the densities of BSs and are independent of the transmit powers of BSs.
 However, the spectral efficiency with CA decreases with the increase of $P_{M}$.
Note that when $P_M$ is as small as~$P_S$, there is not much difference between MBSs and SBSs. Hence, the association probabilities with DA and CA are nearly the same. Consequently, the spectral efficiencies with DA and CA are almost the same.
On the other hand, the probability of associating to MBS in DL increases with the increase of~$P_M$, leading to much severer interference at SBSs and much lower spectral efficiency of users associated to SBSs. As a result, the spectral efficiency with CA decreases.

As a final remark, it is shown in above figures that DA significantly improves the spectral efficiency in UL, especially when the transmit power of MBSs is much higher than that of SBSs.
\section{Conclusion}\label{section6}
In this paper, the problem of cell association in a two-tier network has been investigated in multiuser MIMO communications. We have developed the expressions for the association probabilities and the lower bounds on the spectral efficiency for both DA and CA scenarios. The main results in this paper show that in the multiuser scenario, DA could bring impressive improvement on the spectral efficiency.
It is worth to notice that only single connection for DL and UL is considered in this paper, while users may connect to multiple BSs in practice. Therefore, multiple connections for each user and related DA may be studied in the future. It would also be of interest to explore how different UL power control settings affect the spectral efficiency of users with DA in multiuser MIMO communications as well.
\begin{appendices}
\section{}\label{apx:A}
Since $\gamma_{{v}}^{({\rm DA})}$ is the expectation of interference plus noise at $v$BS, we have
\begin{eqnarray}\label{inr}
  &&\hspace{-0.6cm}\gamma_{{v}}^{({\rm DA})}\triangleq\mathbb{E}_{{I_v^{(U)},\mathbf{a}_{v\!,k}}}\!\!\left[I_v^{(U)}+\left\lVert\mathbf{a}_{v,k}\right\rVert^2\sigma^2\right]=\mathbb{E}_{{I_v^{(U)},\mathbf{a}_{v\!,k}}}\!\!\!\left[\sum_{{j\in\Phi_u\backslash \left\{\!K_v^{(U)}\!\right\}}}\hspace{-0.2cm}Qr_{v,j}^{-\alpha}\left\lvert\mathbf{a}_{v,k}^H\mathbf{h}_{v,j}\right\rvert^2\!+\!\left\lVert\mathbf{a}_{v,k}\right\rVert^2\sigma^2\right]\notag\\
   &&\hspace{-0.4cm}=\!\mathbb{E}_{{r_{v,j}}}\!\!\left[\sum_{{j\in\!\Phi_u\backslash\ \left\{\!K_v^{(U)}\!\right\}}}\!\!\!Qr_{v,j}^{-\alpha}\mathbb{E}\!\left[\left\lVert\mathbf{a}_{v,k}^H\right\rVert^2\!\!\left\lVert Y\right\rVert^2\right]\!\right]\!+\!\mathbb{E}\left[\left\lVert\mathbf{a}_{v,k}\right\rVert^2\right]\!\sigma^2,
\end{eqnarray}
where $Y= \frac{\mathbf{a}_{v,k}^H\mathbf{h}_{v,j}}{\left\lVert\mathbf{a}_{v,k}^H\right\rVert}$. According to \cite{6082486}, the symbol $Y$ is Gaussian distributed with zero mean and unit variance, i.e.,~$Y\!\!\!\sim\!\mathcal{CN}(0,1)$. Moreover, it is independent of $\mathbf{a}_{v,k}^H$. Hence,~it~is~obtained~that
\begin{eqnarray}\label{70_1}
  \gamma_{{v}}^{({\rm DA})}\hspace{-0.4cm}&=&\hspace{-0.3cm}\left(\mathbb{E}_{{r_{v,j}}}\!\!\left[\sum_{{j\in\Phi_u\!\backslash\! \left\{\!K_v^{(U)}\!\right\}}}\!\!\!\!\!\!Qr_{v,j}^{-\alpha}\right]\!\!+\!\sigma^2\right)\mathbb{E}\left[\left\lVert\mathbf{a}_{v,k}\right\rVert^2\right].
\end{eqnarray}
 Since $\left\lVert\mathbf{a}_{v,k}\right\rVert^2=\left[(\mathbf{G}_{v}^H\mathbf{G}_{v})^{-1}\right]_{kk}$, the expectation of $\left\lVert\mathbf{a}_{v,k}\right\rVert^2$ is given as
\begin{equation}\label{75}
  \mathbb{E}\left[\left\lVert\mathbf{a}_{v,k}\right\rVert^2\right] =  \frac{1}{K_v^{(U)}}\mathbb{E}\left[{\rm tr}\left[(\mathbf{G}_{v}^H\mathbf{G}_{v})^{-1}\right]\right].
\end{equation}
According to \cite{tulino2004random}, the term $\mathbf{G}_{v}^H\mathbf{G}_{v}\sim \mathcal{W}_{K_v^{(U)}}(L_v,{\mathbf{I}}_{L_v})$ is a~$K_v^{(U)}\!\!\times K_v^{(U)}$ complex central Wishart matrix with $L_v$ degrees of freedom where $L_v>K_v^{(U)}$, then, it is derived as
\begin{equation}\label{76_1}
  \mathbb{E}\left[{\rm tr}\left((\mathbf{G}_{v}^H\mathbf{G}_{v})^{-1}\right)\right]=\frac{K_v^{(U)}}{L_v-K_v^{(U)}}.
\end{equation}

Furthermore, note that the sum of
interference at $v$BS consists of two parts, one is interference from users whose nearest $v$BS is not $\mbox{MBS}_m$ (or $\mbox{SBS}_n$), the other is interference from users whose nearest $v$BS is $\mbox{MBS}_m$ (or $\mbox{SBS}_n$) but are associated to SBS (or MBS). Therefore, the sum of interference can be rewritten as
\begin{equation}\label{77}
 \sum_{{j\in\Phi_u\backslash \left\{\!K_v^{(U)}\!\right\}}}\hspace{-0.2cm}Qr_{v,j}^{-\alpha}=Q\overline{r}_{{v}}\!+\kappa_v^{(U)}\!Q\overline{r}_{{v},1}^{({\rm DA})},
\end{equation}
where $\kappa_v^{(U)}=\sum_{n=1}^N K_{S,m,n}^{(U)}$ for $v=M$ and $\kappa_v^{(U)}=K_{M,m,n}^{(U)}$ for $v=S$.
The symbol $\overline{r}_{{v}}$ in the first term is derived as
\begin{eqnarray}
  \hspace{-0.8cm}\overline{r}_{{v}}\hspace{-0.3cm}\!&=&\hspace{-0.4cm}\!\int\hspace{-0.1cm}r^{-\!\alpha} \!\!\left(F_{r_{v,k}}(r)\!-\!F_{r_{v,k}}(r_p)\!\right)\!\left(\!1\!-\!F_{r_{w,k}}(r_p)\!\right)\!\lambda_u{\rm d}S\!=\!\!\!\int_{r_p}^{R}\hspace{-0.2cm}r^{-\!\alpha} 2\pi\lambda_ur\left(e^{-\lambda_v\pi r_p^2}\!-\!e^{-\lambda_v\pi r^2}\!\right)\!e^{-\lambda_w\pi r_p^2}{\rm d}r.
\end{eqnarray}
Besides, as shown in (\ref{48}), the CDF utilized to calculate $\overline{r}_{{v},1}^{({\rm DA})}$ is developed as
\begin{eqnarray}
  \hspace{-0.1cm}F_{v}^{({\rm DA})}\!(r)\!=\!\!\frac{{\rm Pr}\!\left(r_p\!<\!r_{v,k}\!<\!r\!<\!R,r_p\!<\!r_{w,k}\!<\!r_{v,k}\!\right)}{A_w^{(U)}}\!=\!\frac{2\pi\lambda_v}{A_w^{(U)}}\!\!\!\int_{r_p}^r\!r_{v\!,k}e^{-\!\lambda_v\pi r_{v\!,k}^2}\!\!\left(\!e^{-\!\lambda_w\pi r_p^2}\!-\!e^{-\!\lambda_w\pi r_{v\!,k}^2}\!\right)\!{\rm d}r_{v\!,k}.
\end{eqnarray}
Then, we derive (\ref{48}) with the corresponding PDF. Finally, by substituting (\ref{75})-(\ref{77}) into (\ref{70_1}), we obtain the desired result in~(\ref{INR51}).
\section{}\label{apx:B}

Let $A1$ represent the event that $K_{v,m,n}^{(U)}$ users are associated to $v$BS in UL and $B1$ denote the event that the~$k$th user is associated to $v$BS in UL. Then, the CDF of the distance from one user to its tagged $v$BS in UL, i.e., $F_{r_{v,k}}^{(U)}(r)$, is given as follows,
\begin{eqnarray}
  && \hspace{-1.4cm}F_{r_{v,k}}^{(U)}\!(r)= {\rm Pr}\left(r_{v,k}<r|B1,A1\right)\notag \\
   &&\hspace{-1cm}= \frac{{\rm Pr}\left(r_p<r_{v,k}<r,r_{w,k}>r_{v,k},r_{w,k}>r_p\right)}{A_v^{(U)}}={\rm exp}\!\left(\!-\pi(\lambda_v\!\!+\!\lambda_w)r_p^2\right)\!\!-\!{\rm exp}\!\left(\!-\pi(\lambda_v\!\!+\!\lambda_w)r^2\right)\!\!.\label{82}
\end{eqnarray}
Note that the expectation of a positive random variable $X$ satisfies~$\mathbb{E}[X]=\int_{t>0}{\rm Pr}\left(X>t\right){\rm d}t$. Hence, the lower bound on the spectral efficiency of users associated to $v$BS can be written as
\begin{eqnarray}\label{85}
 {\rm SE}^{({\rm DA})}_{v,m,n}\!\!\ge\!\! \int_0^{+\infty}\!\!\!{\rm Pr}\!\left(\!{\rm ln}\left(\!1+\frac{Qr_{v,k}^{-\alpha}}{\gamma_{{v}}^{({\rm DA})}}\right)\!>\!t\!\right)\!{\rm d}t=\!\!\!\int_0^{{\rm ln}\left(1+\frac{Q}{\gamma_{{v}}^{({\rm DA})}r_p^{\alpha}}\right)} \!F_{r_{v,k}}^{(U)}\left(\left(\frac{Q}{(e^t-1)\gamma_{{v}}^{({\rm DA})}}\right)^{\frac{1}{\alpha}}\right){\rm d}t.
\end{eqnarray}
The upper limit of the integral in (\ref{85}) is derived by considering the fact that the variable $r$ in~(\ref{82}) satisfies $r>r_p$, resulting in the inequality as follows,
\begin{equation}\label{86}
  \left(\frac{Q}{(e^t-1)\gamma_{{v}}^{({\rm DA})}}\right)^{\frac{1}{\alpha}}>r_p.
\end{equation}
Finally, by using the CDF of $r_{v,k}$ in (\ref{82}) and the lower bound in (\ref{85}), the result in (\ref{51M}) is~derived.
\section{}\label{apx:C}
Let $A2$ be the event that $K_{M,m,n}^{(D)}$ users are associated to $\mbox{MBS}_m$ in DL and $B2$ be the event that the $k$th user is associated to $\mbox{MBS}_m$ in DL, the CDF of the distance from one user to its tagged MBS in DL, i.e., $F_{r_{M,k}}^{(D)}(r)$, is given by
\begin{eqnarray}
  \hspace{-0.8cm}F_{r_{M\!,k}}^{(D)}\!(r)\hspace{-0.3cm}\! &=&\!\!\hspace{-0.2cm}{\rm Pr}\!\left(\!r_{M,k}\!\!<\!r|B2,\!A2\!\right)\!\!=\!{\frac{{\rm Pr}\!\left(\!r_p\!<\!\!r_{M\!,k}\!<\!r,\!r_{S\!,k}\!>\!C_{m\!,n}r_M,\!r_{S\!,k}\!>\!r_p\right)}{A_{M,m,n}^{(D)}}}.
\end{eqnarray}
From (\ref{37}), it can be noted that since $P_M$ is much higher than~$P_S$ and $L_M$ is much larger than~$L_S$, the constant $C_{m,n}$ satisfies $C_{m,n}<~1$. Therefore, the CDF is obtained as
\begin{eqnarray}
F_{r_{M,k}}^{(D,1)}(r)&=&\frac{\widetilde{\lambda}}{\lambda_M}\left(\!e^{-\pi\lambda_Mr_p^2}\!-\!e^{-\pi\lambda_Mr^2}\right)e^{-\pi\lambda_Sr_p^2}, \;r_p<r<\frac{r_p}{C_{m,n}}\label{CDF1}\\
   F_{r_{M,k}}^{(D,2)}(\!r\!)&=& e^{-\widetilde{\lambda}\pi \left(\frac{r_p}{C_{m,n}}\right)^2}-e^{-\widetilde{\lambda}\pi r^2},\,r>\frac{r_p}{C_{m,n}},\label{CDF2}
\end{eqnarray}
where $\widetilde{\lambda}=\lambda_M\!+\!\lambda_SC_{m,n}^2$.

Then, the lower bound on the spectral efficiency of users associated to $\mbox{MBS}_m$ with CA is calculated as
\begin{eqnarray}\label{87}
 \hspace{-0.8cm} {\rm SE}^{({\rm CA})}_{M,m,n}\hspace{-0.3cm} &\geq& \hspace{-0.3cm}\int_0^{+\infty}\hspace{-0.2cm}{\rm Pr}\left(r_{M,k}<\left(\frac{Q}{(e^t-1)\gamma_{{M}}^{({\rm CA})}}\right)^{\frac{1}{\alpha}}\right){\rm d}t\notag\\
 &=&\hspace{-0.3cm}\int_0^{L_0}\hspace{-0.1cm} F_{r_{M,k}}^{(D,2)}\left(\left(\frac{Q}{(e^t-1)\gamma_{{M}}^{({\rm CA})}}\right)^{\frac{1}{\alpha}}\right){\rm d}t+ \!\int_{L_0}^{L_1}  F_{r_{M,k}}^{(D,1)}\left(\left(\frac{Q}{(e^t-1)\gamma_{{M}}^{({\rm CA})}}\right)^{\frac{1}{\alpha}}\right){\rm d}t,
\end{eqnarray}
where $L_0$ and $L_1$ are integration limits. Note that $r$ varies from $r_p$ to $\frac{r_p}{C_{m,n}}$ in (\ref{CDF1}) and $r\!>\!\frac{r_p}{C_{m,n}}$ in (\ref{CDF2}). Then, we have
\begin{eqnarray}
  \frac{Q}{\left(e^{L_0}-1\right)\gamma_M^{({\rm CA})}}&=&\frac{r_p}{C_{m,n}},\label{94_1}\\
  \frac{Q}{\left(e^{L_1}-1\right)\gamma_M^{({\rm CA})}}&=&r_p.\label{94_2}
\end{eqnarray}
Based on (\ref{94_1}) and (\ref{94_2}), the expressions for $L_0$ and $L_1$ in~(\ref{71_1}) and (\ref{71_2}) are derived.
Finally, by plugging (\ref{CDF1}) and (\ref{CDF2}) into~(\ref{87}), the desired result in (\ref{SE_M^CA}) is obtained.

\end{appendices}
\bibliographystyle{IEEEtran}
\bibliography{ref}

\end{document}